\documentclass[11pt, letterpaper]{amsart}
\usepackage{graphicx, amssymb, color}
\usepackage{amsmath}
\usepackage{graphicx}
\usepackage{enumerate}
\usepackage{dsfont}

\newtheorem{thm}{Theorem}[section]

\newtheorem{lem}[thm]{Lemma}
\newtheorem{prop}[thm]{Proposition}
\theoremstyle{definition}

\newtheorem{ass}[thm]{Assumption}

\theoremstyle{remark}
\newtheorem{rem}[thm]{Remark}
\newtheorem{exa}[thm]{Example}

\newtheorem{stat}[thm]{Statement}
\numberwithin{equation}{section}
\newcommand{\norm}[1]{\left\Vert#1\right\Vert}

\newcommand{\Real}{\mathbb R}
\newcommand{\reals}{\mathbb R}

\newcommand{\F}{\mathcal{F}}

\newcommand{\prob}{\mathbb{P}}

\newcommand{\expec}{\mathbb{E}}

\newcommand{\fF}{\mathfrak{F}}

\newcommand{\cL}{\mathcal{L}}

\newcommand{\cZ}{\mathcal{Z}}

\newcommand{\pare}[1]{\left(#1\right)}
\newcommand{\bra}[1]{\left[#1\right]}

\newcommand{\qprob}{\mathbb{Q}}

\newcommand{\nada}[1]{}
\newcommand{\trace}[1]{\textrm{Tr}\left(#1\right)}
\newcommand{\cbra}[1]{\left\{#1\right\}}

\newcommand{\sdpos}{\mathbb{S}_{++}^d}
\newcommand{\sd}{\mathbb{S}^d}
\newcommand{\md}{\mathbb{M}^d}

\newcommand{\mdim}[2]{\mathbb{M}^{#1\times #2}}
\newcommand{\Llam}{\Lambda\Lambda'}
\newcommand{\dfn}{\, := \,}

\newcommand{\W}{\mathcal{W}}

\newcommand{\uk}{\underline{\kappa}}
\newcommand{\ok}{\overline{\kappa}}
\newcommand{\eps}{\varepsilon}
\newcommand{\idmat}[1]{\mathds{1}_{#1}}
\newcommand{\ito}{It\^{o}}

\begin{document}

\title{Long term Optimal Investment in Matrix Valued Factor Models}
\date{\today}

\author[]{Scott Robertson}
\address[Scott Robertson]{Department of Mathematical Sciences,
Wean Hall 6113,
Carnegie Mellon University,
Pittsburgh, PA 15213,
USA}
\email{scottrob@andrew.cmu.edu}

\author[]{Hao Xing}
\address[Hao Xing]{Department of Statistics,
London School of Economics and Political Science,
10 Houghton st,
London, WC2A 2AE,
UK}
\email{h.xing@lse.ac.uk}

\begin{abstract}
Long term optimal investment problems are studied in a factor model with matrix valued state variables. Explicit parameter restrictions are obtained under which, for an isoelastic investor, the finite horizon value function and optimal strategy converge to their long-run counterparts as the investment horizon approaches infinity. This convergence also yields portfolio turnpikes for general utilities. By using results on large time behavior of semi-linear partial differential equations, our analysis extends affine models, where the Wishart process drives investment opportunities, to a non-affine setting. Furthermore, in the affine setting, an example is constructed where the value function is not exponentially affine, in contrast to models with vector-valued state variables.
\end{abstract}


\maketitle

\section{Introduction}

When investment opportunities are stochastic and the market is incomplete, optimal strategies in portfolio choice problems rarely admit explicit forms. The main source of difficulty is that the hedging demand depends implicitly upon the investment horizon. This difficulty motivates  approximating optimal policies, and one useful approximation occurs by considering the long run limit. This approximation enables tractability for optimal strategies and illuminates the relationship between investor preferences, underlying economic factors and dynamic asset demand.  Long run approximations typically take two forms: first, the \emph{long run optimal investment} or \emph{risk sensitive control} problem seeks to identify growth optimal policies for isoelastic utilities; second, the \emph{portfolio turnpike} problem seeks to connect optimal policies for general utilities with those for a corresponding isoelastic utility.

In this article, long run optimal investment and portfolio turnpike problems are studied in a multi-asset factor model where the state variable takes values in the space of positive definite matrices. Such models generalize the Wishart model of \cite{buraschi2010correlation, Hata-Sekine} (amongst many others), which has been successfully employed in a wide-range of problems in Mathematical Finance. In addition to identifying optimal long run policies and proving turnpike theorems, we are particularly concerned with connecting the finite horizon and long run problems. Here, the goal is to provide conditions when optimal policies for finite horizons converge to their long-run counterparts. Positive results in this direction are necessary to validate long-run analysis. Though heuristics indicate convergence, from a technical standpoint it is not a priori clear that the long-run policy arises as the limit of finite horizon policies.

For isoelastic utilities, the risk sensitive control, or long run optimal investment, problem aims to maximize the expected utility growth rate. This problem has been addressed by many authors : see, for example, \cite{MR1675114, MR1790132, MR1882297, MR1802598, MR1910647, MR1932164, MR1995925, Nagai-03, MR2435642, Guasoni-Robertson, MR2477847}. In these studies,  an ergodic Hamilton Jacobi Bellman (HJB) equation is analyzed. This ergodic equation is typically obtained via a heuristic argument, where one first derives the finite horizon HJB equation, and then conjectures that for long horizons the (reduced) value function decomposes into the sum of a spatial component and a temporal growth component.  Thus, if $v(T,\cdot)$ denotes the finite horizon value function, the long-run value function takes the form $\hat{\lambda}T + \hat{v}(\cdot)$. Then ergodic HJB equation follows by substituting the latter function into the finite horizon HJB equation.

The above heuristic derivation indicates that finite and infinite horizon optimal investment problems are parallel in many aspects. Of primary importance is to connect these two class of problems. As the investment horizon $T$ approaches infinity, does the finite horizon value function $v(T, \cdot)$ converge to its long-run analogue $\hat{\lambda} T + \hat{v}(\cdot)$? If so, in what sense? Does the optimal strategy for the finite horizon problem converge to a long-run limit? As previously mentioned, affirmative answers to these questions verify the intuition underpinning the study of the risk sensitive controls, and provide consistency between the finite horizon and long-run problems.

Moving away from the isoelastic case,  portfolio turnpikes provide another approximation for optimal policies of generic utility functions.  Qualitatively, turnpike theorems state that in a growing market (i.e. one where the riskless asset tends to infinity), as the investment horizon becomes large, the optimal trading strategy of a generic utility converges, over any finite time window, to the optimal trading strategy of its isoelastic counterpart (see Assumption \ref{ass: ratio} for a precise formulation of ``counterpart'').  Turnpike theorems were first investigated in \cite{mossin1968optimal} for utilities with affine risk tolerance, and have since been extensively studied: in particular we mention \cite{leland1972turnpike, ross1974portfolio, hakansson1974convergence, MR736053, MR1172445, MR1629559, MR1805320, dybvig1999portfolio, detemple2010dynamic} where turnpike theorems are proved in differing levels of generality.

For the risk-sensitive control and turnpike approximations, we summarize the relationship between the finite and long horizon problems in Statements \ref{stat: long hor} and \ref{stat: turnpike} respectively. Verification of these statements allows investors with a long horizon to replace their optimal, but implicit, strategies with explicit long-run approximations, which lead to minimal loss of their wealth and utility, while providing considerable tractability. Each of Statements \ref{stat: long hor} and \ref{stat: turnpike} have been proved in \cite{guasoni.al.11} in a factor model with univariate state variable and constant correlation of hedgeable and unhedgeable shocks. The present paper extends these results to a multivariate setting, which allows for stochastic interest rates, volatility, and correlation.  Here, in our main results, Proposition \ref{prop: wishart} and Theorems \ref{thm: power}, \ref{thm: turnpike}, we provide explicit parameter assumptions upon the model coefficients under which both Statements \ref{stat: long hor} and \ref{stat: turnpike} hold.

As previously stated, we focus on a factor model where the state variable is matrix valued. This is motivated by consideration of the Wishart process (cf. \cite{Bru} and Example \ref{ex: wishart} below), which has been applied to option pricing (cf. \cite{Gourieroux06, Gourieroux09, DaFonseca08, DaFonseca10}). Its application to portfolio optimization was pioneered by \cite{buraschi2010correlation}, which highlighted the impact of the multivariate state variable on the hedging demand. In particular, using practical relevant parameters, the numerical example in Section B.3 therein showed that the hedging demand converges to a steady-state level when the investment horizon is longer than $5$ years. Our results confirm this observation. In \cite{Hata-Sekine}, the portfolio optimization problem is solved in the Wishart case via a matrix Riccati differential equation. In \cite{Bauerle-Li}, logarithmic utility is studied, and in \cite{Richter} the indifference pricing is discussed.

In contrast to the aforementioned results, which exploit the affine structure of the Wishart process, our results rely upon large time asymptotic analysis of partial differential equations with quadratic nonlinearities in the gradient. Using techniques developed in \cite{Robertson-Xing}, we are able to consider non-affine models, and hence discuss general matrix-valued state variables as in Section \ref{subsec: state var}. Moreover, stochastic correlation between the state variable and risky assets can be treated, whereas a special (constant) correlation structure is needed to ensure the affine structure. Furthermore, our analysis, when applied to affine models, yields new insight: we construct a counter-example (Example \ref{exa: counter-exa}) to the long-held belief that optimal policies are affine in affine models.  Indeed, the model in this example is affine, but the associated value function is not exponentially affine, hence the optimal policy is not affine. This happens when the dimension of state variable is larger than the number of risky assets, and is due to the noncommuntative property of the matrix product.

The paper is organized as follows: after the model and Statements \ref{stat: long hor} and \ref{stat: turnpike} are introduced in Section \ref{sec: setup}, the main results are presented in Section \ref{sec: converge}. For ease of exposition, the general results are first specified to when the state variable follows a Wishart process in Section \ref{subsec: wishart}. Here, the investment model may or may not be affine depending upon the asset drifts and covariances. Proposition \ref{prop: wishart} provides simple, mild (especially in the case where the investor risk aversion exceeds that of a logarithmic investor) parameter restrictions under which the main results follow. Proposition \ref{prop: wishart_good_case} explicitly identifies the long-run limit policy when the model is further specified to the ``classical'' affine Wishart model considered in \cite{buraschi2010correlation, Hata-Sekine} and Example \ref{exa: counter-exa} constructs the non exponentially affine counter example.  After considering the Wishart case, the main results for general matrix valued state variables are given in Section \ref{subsec: general_case} : see Theorem \ref{thm: power} for the long run limit results and Theorem \ref{thm: turnpike} for the turnpike results.  All proofs are deferred to Appendices \ref{app: A}, \ref{app: C} and \ref{app: B}. Finally, we summarize several notations used throughout the paper:

\begin{itemize}
 \item $\mdim{d}{k}$ denotes the space of $d\times k$ matrices with $\md\dfn \mdim{d}{d}$. For $x\in \mdim{d}{k}$, denote by $x'$ the transpose of $x$. For $x\in \md$, denote by $\trace{x}$ the trace of $x$ and $\norm{x} = \sqrt{\trace{x'x}}$. For $x,y\in \md$, the Kronecker product of $x$ and $y$ is denoted by $x\otimes y\in \mathbb{M}^{d^2}$. Denote by $\idmat{d}$ the identity matrix in $\md$ and $1_d$ the $d$-dimensional vector with each component $1$.
 \item $\sd$ denotes the space of $d\times d$ symmetric matrices, and $\sdpos$ the cone of positive definite matrices. For $x\in\sdpos$, denote by $\sqrt{x}$ the unique element $y\in\sdpos$ such that $y^2 = x$. For $x, y\in \sdpos$, $x\geq y$ when $x-y$ is positive semi-definite.
 \item For $E\subset \mdim{d}{k}$, $F\subset \mdim{m}{n}$, and $\gamma\in(0,1]$, denote by $C^{\ell,\gamma}(E; F)$ the space of $\ell$ times continuously differentiable functions from $E$ to $F$ whose derivatives of order up to $\ell$ is locally H\"{o}lder continuous with exponent $\gamma$.
\end{itemize}

\section{Set up}\label{sec: setup}
Let $(\Omega, (\F_t)_{t\geq 0}, \F, \prob)$ be a filtered probability space with $(\F_t)_{t\geq 0}$ a right-continuous filtration. Following the treatment in \cite{guasoni.al.11}, all $N$-negligible sets (cf. \cite[Definition 1.3.23]{Bitchteler} and \cite{N-N}) are included into $\F_0$. Such a completion of $\F_0$ ensures, for all $T\geq 0$, that  $(\Omega, (\F_t)_{0\leq t\leq T}, \F_T, \prob)$ satisfies the usual conditions.

Consider a financial model with one risk-free asset $S^0$ and $n$ risky assets $(S^1,...,S^n)$. Investment opportunities are driven by a $\sdpos$ valued state variable $X$.  Before writing down the dynamics for the assets, it is necessary to introduce the state variable $X$, as the dynamics for $X$ involve matrix notation.

\subsection{A $\sdpos$-valued state variable}\label{subsec: state var}
Let $B= (B^{ij})_{i,j=1,...d}$ be a $\md$-valued Brownian motion on $(\Omega, (\F_t)_{t\geq 0}, \F, \prob)$. The state variable $X$ has dynamics
\begin{equation}\label{eq: state}
dX_t = b(X_t)dt + F(X_t)dB_t G(X_t) + G(X_t)'dB'_t F(X_t)',\qquad X_0\in\sdpos.
\end{equation}
Here, $b\in C^{1,\gamma}(\sdpos; \sd)$ and $F, G \in C^{2,\gamma}(\sdpos; \md)$ are given functions. We require $b,F,G$ to be such that $X$ possesses a unique strong solution which is non-explosive, i.e.,
\begin{equation*}\label{eq: state_no_expl}
\prob^x\bra{X_t \in \sdpos,\ \forall \ t\geq 0} =1, \qquad \text{ for all  } x\in \sdpos,
\end{equation*}
where $\prob^x$ is the probability such that $X_0= x$ a.s..
To enforce this requirement through restrictions upon $b,F$ and $G$, the results as well as notation of \cite{Mayerhofer-Pfaffel-Stelzer} are used. Namely, define
\begin{equation}\label{eq: f_g_def}
f(x):= FF'(x) \quad \text{and} \quad g(x):= G'G(x),\qquad x\in \sdpos.
\end{equation}
Next, given $b,f,g: \sdpos\rightarrow \sd$ and $\delta \in \Real$, define $H_\delta: \sdpos \rightarrow \Real$ via
\begin{equation}\label{eq: Hdelta_def}
 H_\delta(x; b) := \trace{b\, x^{-1}} - (1+
 \delta)\,\trace{fx^{-1}g x^{-1}} - \trace{f\,x^{-1}}\,
 \trace{g\, x^{-1}}, \qquad x\in\sdpos.
\end{equation}
Here, we have omitted the function arguments from $b,f,g$ but have explicitly identified the drift function $b$ in $H_\delta$, since in the sequel $H_\delta$ will be used with various $b$.

To understand $H_\delta$, note that if $X$ from \eqref{eq: state} has a strong solution satisfying \eqref{eq: state_no_expl} then \ito's formula implies the drift in the dynamics for $\log(\det(X_t)))$ is $H_0(X_t;b)$. Thus, the following assumption ensures that $X$ from \eqref{eq: state} neither explodes in norm nor has degenerate determinate and hence possesses a unique global strong solution $(X_t)_{t\in \Real_+}$ on $\sdpos$, cf. \cite[Theorem 3.4]{Mayerhofer-Pfaffel-Stelzer}.

\begin{ass}\label{ass: wp}\text{}
 \begin{enumerate}[i)]
 \item $G'\otimes F$ and $b$ are locally Lipschitz and of linear growth.
 \item $\inf_{x\in \sdpos} H_0(x; b)>-\infty$.
 \end{enumerate}
\end{ass}

\begin{rem}\label{rem: kron}
A direct calculation, using \cite[Section 4.2]{Horn-Johnson}, shows that
\begin{equation*}
\begin{split}
\|G'\otimes F(x) - G'\otimes F(y)\|^2 &\leq 2\left(\|G(x)\|^2\|F(x)-F(y)\|^2 + \|F(y)\|^2\|G(x)-G(y)\|^2\right),\\
\|G'\otimes F(x)\|^2 &= \|F(x)\|^2\|G(x)\|^2 = \trace{f}\trace{g},\qquad \text{ for } x,y\in\sdpos.
\end{split}
\end{equation*}
Thus, $G'\otimes F$ will be locally Lipschitz and of linear growth once $F$ and $G$ are locally Lipschitz and $\|F(x)\| \|G(x)\|\leq C(1+\|x\|)$ or equivalently if $\trace{f}\trace{g}\leq C(1+\|x\|^2)$.

\end{rem}

Assumption \ref{ass: wp} establishes well-posedness of \eqref{eq: state}. The next assumption implies that the volatility of $X$ is non-degenerate in the interior of $\sdpos$.
\begin{ass}\label{ass: loc_ellip}
For each $x\in\sdpos$, $f(x) > 0$ and $g(x) > 0$.
\end{ass}
Indeed, note that \eqref{eq: state} is short-hand for the following system:
\begin{equation*}
dX^{ij}_t = b_{ij}(X_t)dt + \sum_{k,l=1}^{d} F(X_t)_{ik}dB^{kl}_tG(X_t)_{lj}+\sum_{k,l=1}^d F(X_t)_{jk}dB^{kl}_tG(X_t)_{li}, \qquad i,j = 1,...,d.
\end{equation*}
For $i,j=1,...,d$ define the matrix $a^{ij}:\sdpos\rightarrow\md$ by
\begin{equation*}\label{eq: a_def}
a^{ij}_{kl}(x) \dfn \left(F_{ik}G_{lj} + F_{jk}G_{li}\right)(x),\qquad k,l = 1,...,d,\,
x\in \sdpos.
\end{equation*}
Then the above system takes the form
\begin{equation*}\label{eq: state_alt}
dX_t^{ij} = b_{ij}(X_t)dt + \trace{a^{ij}(X_t)dB'_t}.
\end{equation*}
Then \cite[Lemma 5.1]{Robertson-Xing} shows that under Assumption \ref{ass: loc_ellip}, for any $x\in \sdpos$ and $\theta\in \sd$,
\begin{equation}\label{eq: ellip}
\sum_{i,j,k,l=1}^d \theta_{ij}\trace{a^{ij}(a^{kl})'}(x) \theta_{kl} = 4 \trace{f(x) \theta g(x) \theta}\geq c(x)\norm{\theta}^2,
\end{equation}
for some constant $c(x)>0$.

\begin{exa}\label{ex: wishart}
The primary example to keep in mind is when $X$ is the Wishart process, cf. \cite{Bru}:
\begin{equation}\label{eq: wishart}
dX_t = \pare{L L' + K X_t + X_t K'} dt + \sqrt{X_t} dB_t \Lambda' + \Lambda
dB'_t \sqrt{X_t},
\end{equation}
where $K,L,\Lambda \in \md$. Then both Assumptions \ref{ass: wp} and \ref{ass: loc_ellip} are satisfied when
\begin{equation}\label{eq: wishart_cond}
LL' \geq (d+1)\Llam > 0.
\end{equation}
Indeed, here $b(x) = LL' +  Kx + xK'$, $f(x) = x$, and $g(x) = \Llam$. Using Remark \ref{rem: kron} it follows that $b,G'\otimes F$ are locally Lipschitz and of linear growth. Furthermore, calculation shows that $H_0(x; b) = \trace{(LL' - (d+1)\Llam)x^{-1}} + 2\trace{K}$. Thus, the first inequality in \eqref{eq: wishart_cond} implies $H_0(x;b) \geq 2\trace{K}$ on $\sdpos$ and Assumption \ref{ass: wp} holds. Assumption \ref{ass: loc_ellip} readily follows from the second inequality in \eqref{eq: wishart_cond}.
\end{exa}

\subsection{The financial model}\label{subsec: fin market}

Having fixed notation and established well-posedness for the state variable, we may now define the financial model.  As mentioned above, there is one risk-free asset $S^0$ and $n$ risky assets $(S^1,...,S^n)$ whose dynamics are given by
\begin{align}
\frac{dS^0_t}{S^0_t} &= r(X_t)dt,\qquad S^0_0 = 1, \label{eq: sde safe}\\
\frac{dS^i_t}{S^i_t} &= \left(r(X_t) + \mu_i(X_t)\right) dt +
   \sum_{j=1}^m \sigma_{ij}(X_t)dZ^j_t,\qquad S^i_0 > 0,\qquad i=1,...,n.\label{eq: sde S}
\end{align}
Here, $r\in C^\gamma(\sdpos; \Real)$, $\mu\in C^{1,\gamma}(\sdpos; \Real^n)$, $\sigma\in C^{2,\gamma}(\sdpos; \mdim{n}{m})$ and $Z = (Z^1,...,Z^m)$ is a $\Real^m$ valued Brownian motion. That $\sigma$ is of full rank, as well as the existence of \emph{market price of risk}, i.e., $\nu: \sdpos \rightarrow \Real^n$ such that $\mu = \sigma \sigma' \nu$ on $\sdpos$, are ensured by the following assumption:
\begin{ass}\label{ass: sig_ellip}\text{}
\begin{enumerate}[i)]
\item When $m > n$, $\Sigma(x)\dfn \sigma\sigma'(x) > 0$ for $x\in \sdpos$. Then $\nu:= \Sigma^{-1} \mu$.
\item When $m < n$, $\sigma'\sigma(x) > 0$ for $x\in\sdpos$ and there exists $\nu\in C^{1,\gamma}(\sdpos; \Real^n)$ such that $\mu = \Sigma\nu$.
\item When $m=n$, $\Sigma(x)>0$ for $x\in\sdpos$ and $\sigma = \sqrt{\Sigma}$. Here again, $\nu = \Sigma^{-1}\mu$.
\end{enumerate}
\end{ass}

To allow for potentially stochastic instantaneous correlations between asset returns and the state variable, we define $Z$ in terms of the Brownian motion $B$ which drives $X$ and an independent $\Real^m$ valued Brownian motion $W$.  Specifically, let $C\in C^{2,\gamma}(\sdpos;\mdim{m}{d})$ and $\rho\in C^{2,\gamma}(\sdpos; \reals^d)$ be such that
\begin{ass}\label{ass: rho}
$\rho'\rho(x) CC'(x) \leq \idmat{m}$ for each $x\in \sdpos$.
\end{ass}
Set $D:=\sqrt{\idmat{m}- \rho'\rho C C'} \in C^{2,\gamma}(\sdpos; \sd)$. We then may define $Z$ by
\begin{equation}\label{eq: BM Z}
 Z^j_t:= \sum_{k,l=1}^d\int_0^tC_{jk}(X_u)dB^{kl}_u\rho_l(X_u) + \sum_{k=1}^m\int_0^tD_{jk}(X_u)dW^k_u,\qquad t\geq 0, j=1,...,m.
\end{equation}
By construction, $Z$ is a $m$ dimensional Brownian motion. Furthermore, the instantaneous correlation between $Z$ and $B$ is $d\langle Z^j, B^{kl} \rangle_t= C_{jk}(X_t) \rho_l(X_t) dt$, for $1\leq j\leq m, 1\leq k,l\leq d$. In particular, when $m=d$, $C= \idmat{d}$ and $\rho\in\Real^d$ is constant, $d\langle Z^i, B^{jl} \rangle_t= \delta_{ij} \rho_l dt$, where $\delta_{ij}=1$ for $i=j$ and $0$ otherwise. This particular correlation structure is assumed in \cite{buraschi2010correlation, Hata-Sekine, Bauerle-Li, Richter}. Here, the matrix $C$ introduces general correlation structure and allow its dependence upon the state variable $X$.

\subsection{The optimal investment problem}\label{subsec: inv prob}
Consider an investor whose preference is described by a utility function $U: \Real_+ \rightarrow \Real$ which is strictly increasing, strictly concave, continuously differentiable and satisfies the Inada conditions $U'(0) =\infty$ and $U'(\infty)=0$. In particular, we pay special attention to utilities with constant relative risk aversion (henceforth CRRA) $U(x)= x^p/p$ for $0\neq p<1$.

Starting from an initial capital, this investor trades in the market until a time horizon $T\in \Real_+$. She puts a proportion of her wealth $(\pi_t)_{t\leq T}$ into the risky assets and the remaining into the risk free asset. Given her strategy $\pi$, the price dynamics in \eqref{eq: sde safe} and \eqref{eq: sde S} imply that the wealth process $\mathcal{W}^\pi$ has dynamics
\begin{equation}\label{eq: wealth_dyn}
\frac{d\W^\pi_t}{\W^\pi_t} = (r(X_t) + \pi_t'\Sigma(X_t)\nu(X_t))dt + \pi_t'\sigma(X_t) dZ_t.
\end{equation}
The set of \emph{admissible} strategies are those $\pi$ which are $\mathbb{F}$-adapted and such that $\prob^x\bra{\W^\pi_t >0, \forall t\leq T} =1$ for all $x\in\sdpos$. In \eqref{eq: M_eta_def} below, positive super-martingale $M$ are constructed such that $M \W^\pi$ is a super-martingale for any admissible strategy $\pi$. In the presence of such \emph{super-martingale deflators}, arbitrage is excluded from the model (cf. \cite{Karatzas-Kardaras}).
The investor seeks to maximize the expected utility of her terminal wealth at $T$ by choosing admissible strategies, i.e.,
\begin{equation}\label{eq: op}
 \expec\bra{U(\mathcal{W}^\pi_T)}\rightarrow \text{Max}.
\end{equation}

In the remainder of this section, we will focus on the optimal investment problem for CRRA utilities and derive the associated HJB equation via a heuristic argument. To this end, define the (reduced) value function $v$ via
\begin{equation}\label{eq: power op}
\sup_{\pi \text{ admissible}} \expec\bra{\left. \frac{1}{p}\pare{\W^\pi_T}^p \right|\W_t=w, X_t=x} = \frac{1}{p} w^p e^{v(T-t, x)},\qquad 0\leq t\leq T, w>0, x\in\sdpos.
\end{equation}
Set $L$ as the infinitesimal generator of \eqref{eq: state_alt}:
\begin{equation}\label{eq: state_gen}
 L := \frac12 \sum_{i,j,k,l=1}^d \trace{a^{ij}(a^{kl})'} D^2_{(ij),(kl)} + \sum_{i,j=1}^d b_{ij} D_{(ij)},
\end{equation}
where $D_{(ij)}=\partial_{x^{ij}}$ and $D^2_{(ij),(kl)}=\partial^2_{x^{ij} x^{kl}}$. The standard dynamic programming argument yields the following HJB equation for $v$:
\begin{equation}\label{eq: v eqn}
\begin{split}
 \partial_t v = &L v + \frac12 \sum_{i,j,k,l=1}^d D_{(ij)} v \trace{a^{ij}(a^{kl})'} D_{(kl)} v
 + p\,r
 \\
 &+ \sup_{\pi} \left\{p \pi' \pare{\Sigma \nu + \sum_{i,j=1}^d \sigma C a^{ij} \rho D_{(ij)} v} + \frac12 p(p-1) \pi' \Sigma \pi\right\}, \qquad t>0, x\in \sdpos,\\
 0= &v(0, x), \quad x\in \sdpos.
\end{split}
\end{equation}
The optimizer $\pi$ in the previous equation can be obtained pointwise and is given by
\begin{equation}\label{eq: pi_v_map}
\pi(t,x; v) \dfn \begin{cases} \frac{1}{1-p}\Sigma^{-1}\left(\Sigma\nu +
    \sum_{i,j=1}^d\sigma C a^{ij}\rho D_{(ij)}v\right)(t,x), & m > n \\
  \frac{1}{1-p}\sigma(\sigma'\sigma)^{-1}\left(\sigma'\nu + \sum_{i,j=1}^d
    C a^{ij}\rho D_{(ij)}v\right)(t,x), & m \leq n\end{cases},\qquad t>0,x\in\sdpos.
\end{equation}
Define $q:=p/(p-1)$ as the conjugate of $p$ and the function $\Theta: \sdpos \rightarrow \sdpos$ via
\begin{equation}\label{eq: Theta_def}
\Theta(x)\dfn  \begin{cases} \sigma'\Sigma^{-1}\sigma(x) & m> n\\ \idmat{m}
  & m \leq n\end{cases},\qquad x\in\sdpos.
\end{equation}
Plugging in the formula for $\pi$ in \eqref{eq: pi_v_map} into \eqref{eq: v eqn}, a lengthy calculation yields the following semi-linear Cauchy problem for $v$:
\begin{equation}\label{eq: v_pde}
\begin{split}
v_t(t,x) &= \fF[v](t,x),\qquad 0< t, x\in\sdpos,\\
v(0,x)& = 0,\qquad x\in\sdpos.
\end{split}
\end{equation}
Here, the differential operator $\fF$ is defined as
\begin{equation}\label{eq: x_ops}
\fF \dfn
\frac{1}{2}\sum_{i,j,k,l=1}^{d}A_{(ij),(kl)}D^2_{(ij),(kl)}
+ \sum_{i,j=1}^{d} \bar{b}_{ij} D_{(ij)}+
\frac{1}{2}\sum_{i,j,k,l=1}^{d}D_{(ij)}\bar{A}_{(ij),(kl)}D_{(kl)}
+ V,
\end{equation}
with
\begin{equation}\label{eq: functions}
\begin{split}
A_{(ij),(kl)}(x) &:= \trace{a^{ij}(a^{kl})'}(x),\\
\bar{A}_{(ij),(kl)}(x) &:=\trace{a^{ij}(a^{kl})'}(x)-q\rho'(a^{ij})'C'\Theta
    C a^{kl}\rho(x),\\
\bar{b}_{ij}(x)&:= b_{ij}(x) - q\nu'\sigma C a^{ij}\rho(x),\\
V(x)&:= pr(x) - \frac{1}{2}q\nu'\Sigma\nu(x), \qquad i,j,k,l = 1,...,d, x\in \sdpos.
\end{split}
\end{equation}
Note that $\pi$ in \eqref{eq: pi_v_map} and $\fF$ in \eqref{eq: x_ops} take different forms depending on $m>n$ or $m\leq n$ (with the two forms coinciding at $m=n$), and that using the definition of $L$ from \eqref{eq: state_gen} we have
\begin{equation}\label{eq: x_ops_L}
\fF = L - q\sum_{i,j=1}^d \nu'\sigma C a^{ij} \rho D_{(ij)} + \frac{1}{2}\sum_{i,j,k,l=1}^{d}D_{(ij)}\bar{A}_{(ij),(kl)}D_{(kl)}
+ V.
\end{equation}

In Section \ref{sec: converge} well-posedness of \eqref{eq: v_pde} is proved under appropriate parameter assumptions, and it is shown that the solution $v$, with appropriate growth constraint, to \eqref{eq: v_pde} is the reduced value function in \eqref{eq: power op}. Moreover the optimal strategy for \eqref{eq: power op} is given by
\begin{equation}\label{eq: opt_strat_T}
   \pi^T_t \dfn \pi(T-t, X_t; v), \qquad 0\leq t\leq T,
\end{equation}
for $\pi(\cdot, \cdot; v)$ from \eqref{eq: pi_v_map}.

\subsection{Long Horizon Convergence}

As mentioned in the introduction, this article is concerned with the large time behavior of the optimal investment problem. Such behavior for a CRRA investor is closely related to the ergodic analog of \eqref{eq: v_pde}, given by
\begin{equation}\label{eq: v_ergodic}
\begin{split}
\lambda & = \fF[v](x), \qquad x\in \sdpos.
\end{split}
\end{equation}
A solution to \eqref{eq: v_ergodic} is defined as a pair $(\lambda, v)$ where $\lambda\in\Real$ and $v\in C^2(\sdpos; \Real)$ which satisfy \eqref{eq: v_ergodic}. Since $\fF[v]$ only depends on derivatives of $v$, $v$ in a solution is only determined up to an additive constant.  In particular we are interested in the \emph{smallest} $\lambda$ such that \eqref{eq: v_ergodic} admits a solution.

In the study of long horizon optimal investment and risk sensitive control problems,  when the state variable is in $E\subseteq \Real^d$, under appropriate restrictions \cite{Ichihara, Guasoni-Robertson}, there does exist a smallest $\hat{\lambda}$ such that \eqref{eq: v_ergodic} has a solution $\hat{v}$, such that the candidate reduced long run value function, accounting for the growth rate, is $\hat{\lambda} T + \hat{v}(x)$. The candidate long run optimal strategy is
\begin{equation}\label{eq: opt_strat}
  \hat{\pi}_t := \pi(X_t; \hat{v}), \qquad t\geq 0,
 \end{equation}
 where $\pi(\cdot; \hat{v})$ from \eqref{eq: pi_v_map} with $v$ replaced by $\hat{v}$ which does not have a time argument. Now when the state variable is matrix valued, Proposition \ref{prop: ergodic wellposed} below establishes the existence of such $(\hat{\lambda}, \hat{v})$.

Comparing the finite and long horizon problems, we are interested in proving the following claim:

\begin{stat}[Long Horizon Convergence]\label{stat: long hor}$\,$
\begin{enumerate}[i)]
 \item Define $h(T, x) := v(T,x)-\hat{\lambda}T - \hat{v}(x)$, for $T\geq 0$ and $x\in \sdpos$. Then
     \[
      h(T, \cdot) \rightarrow C \quad \text{ and } \quad \nabla h(T,\cdot) \rightarrow 0 \quad \text{ in } C(\sdpos), \quad \text{ as } T\rightarrow \infty.
     \]
     Here $C$ is a constant, $\nabla=(D_{(ij)})_{1\leq i,j\leq d}$ is the gradient operator, and convergence in $C(\sdpos)$ stands for locally uniformly convergence in $\sdpos$.
 \item As functions of $x\in\sdpos$ the finite horizon strategies converge to the long-run counterpart, i.e.
     \[
      \lim_{T\rightarrow \infty} \pi(T, \cdot; v) = \pi(\cdot; \hat{v}) \quad \text{ in } C(\sdpos).
     \]
 \item Let $\pi^T$ and $\hat{\pi}$ be as in \eqref{eq: opt_strat_T} and \eqref{eq: opt_strat}. Let $\mathcal{W}^T$ and $\hat{\mathcal{W}}$ be the wealth processes employing $\pi^T$ and $\hat{\pi}$ respectively starting with initial capital $w$. Then for all $x\in\sdpos$ and all $t\geq 0$:
 \begin{align}
  &\prob^x-\lim_{T\rightarrow \infty} \sup_{0\leq u\leq t} \left|\frac{\mathcal{W}^T_u}{\hat{\mathcal{W}}_u} -1\right| =0,\label{eq: ratio wealth power}\\
  &\prob^x-\lim_{T\rightarrow \infty} \int_0^t (\pi^T_u- \hat{\pi}_u)' \Sigma(X_u) (\pi^T_u - \hat{\pi}_u) \, du =0.\label{eq: dist strategy}
 \end{align}
 Here $\prob^x-\lim$ stands for convergence in probability $\prob^x$.
 \end{enumerate}
\end{stat}

In Statement \ref{stat: long hor}, i) claims that the reduced value function for the finite horizon problem converges to its infinite horizon counterpart; moreover ii) indicates that the finite horizon optimal strategy also converges, in feedback form, to a myopic long run limit. In addition to these analytic results, iii) states convergence in probabilistic terms: that is, the ratio between optimal wealth processes and distance between optimal strategies, when measured in a finite time window $[0,t]$,  converge to zero in probability. Therefore when Statement \ref{stat: long hor} holds, a CRRA investor with long horizon can slightly modify her optimal strategy $\pi^T$ to $\hat{\pi}$, at the beginning of investment period, and incur a minimal loss of wealth and utility. Indeed, under appropriate parameter assumptions, Statement \ref{stat: long hor} is proved in \cite{guasoni.al.11} when the state variable is $\Real$ valued and has constant correlation with risky assets. In Section \ref{sec: converge} below, we will verify Statement \ref{stat: long hor} in the matrix setting.

\subsection{Turnpike Theorems}

To state turnpike results, we consider two investors: the first one has a general utility function $U$ which satisfies conditions at the beginning of Section \ref{subsec: inv prob}; the second investor has a CRRA utility $U(x) = x^p/p$ for $0\neq p < 1$ \footnote{The logarithmic utility case is excluded here, since \cite[Proposition 2.5]{guasoni.al.11} already shows that turnpike theorems hold in a general semimartingale setting including the current case.}. The two investors are connected through the ratio of their marginal utilities $U'(x)/x^{p-1}$ as in the following assumption:

\begin{ass}\label{ass: ratio}
With $\mathfrak{R}(x)\dfn U'(x)/x^{p-1}$ it follows that
\begin{equation}\label{eq: conv}
\lim_{x\uparrow\infty} \mathfrak{R}(x) = 1.
\end{equation}
\end{ass}

Assumption \ref{ass: ratio} ensures that preferences of the two investors are similar for large wealths.  The next assumption ensures that the market described in Section \ref{subsec: fin market} is growing over time.

\begin{ass}\label{ass: grow}
For $r(x)$ as in \eqref{eq: sde safe} there exits constants $0 < \underbar{r} < \bar{r}$ such that $\underbar{r}\leq r(x) \leq \bar{r}$ for all $x\in\sdpos$.
\end{ass}

In order to present the turnpike results, for the investor with general utility $U$, set $\pi^{1,T}$ as the optimal strategy of \eqref{eq: op} and $\W^{1,T}$ as the associated optimal wealth process starting from initial wealth $w$. We are interested in proving the turnpike theorem:

\begin{stat}[Turnpike Theorem]\label{stat: turnpike}
 For all $x\in\sdpos$ and all $t\geq 0$,
\begin{align}
&\prob^x-\lim_{T\rightarrow\infty} \ \sup_{u\leq t}\left|\frac{\W^{1,T}_u}{\hat{\W}_u} - 1\right| = 0,\label{eq: turnpike wealth}\\
&\prob^x-\lim_{T\rightarrow\infty} \int_0^t \left(\pi^{1,T}_u - \hat{\pi}_u\right)'\Sigma(X_u)\left(\pi^{1,T}_u - \hat{\pi}_u\right) du = 0, \label{eq: turnpike strat}
\end{align}
where $\hat{\pi}$ from \eqref{eq: opt_strat} and $\hat{\W}$ is the wealth process starting from $w$ following $\hat{\pi}$.
\end{stat}

The first convergence above states that the ratio, when measured in an finite time window, of the optimal wealth process for the generic investor and the long run wealth process for the CRRA investor is uniformly close to one in probability as the horizon becomes large. The message behind the second convergence is that, as the horizon becomes long, the optimal investment strategy for the generic utility investor approaches the long-run limit strategy of the CRRA investor.
Such a result is called an ``explicit" turnpike using the terminology of \cite{guasoni.al.11}, where Statement \ref{stat: turnpike} is proved in a factor model with $\Real$ valued state variable and constant correlation. In Section \ref{sec: converge} below, we will extend this result to when the state variable is matrix valued.

\begin{rem}\label{rem: general state space}
 Statements \ref{stat: long hor} and \ref{stat: turnpike} are not specific to models with matrix valued state variables. As mentioned in introduction, the main technique to confirm these statements is  the large time asymptotic analysis of \eqref{eq: v eqn} in \cite{Robertson-Xing}. In particular, a general framework is introduced in \cite[Section 2]{Robertson-Xing}, where convergence results (cf. Theorems 2.9 and 2.11 therein) are obtained for a general state space $E$. The main message therein is, when two ``Lyapunov'' functions $\phi$ and $\psi$ exist and satisfy appropriate assumptions, then the desired convergence results hold. When the state space is specified, assumptions on $\phi$ and $\psi$ are translated to explicit parameter restrictions. In particular, when the state space is $\Real^d$, these parameter restrictions are given in \cite[Section 3.1]{Robertson-Xing}. Therefore, proof of Statements \ref{stat: long hor} and \ref{stat: turnpike} in this case follows from essentially the same line of reasoning as in the matrix case and is, in fact, much more straightforward.
\end{rem}

\section{Main results}\label{sec: converge}

\subsection{The (generalized) Wishart factor model}\label{subsec: wishart}
Before presenting results for the general matrix setting in Section \ref{subsec: state var}, let us highlight the case when $X$ is a Wishart process as in Example \ref{ex: wishart}. We specify the financial model in Section \ref{subsec: fin market} to the following:
\begin{align*}
 & m=d, \quad C(x) = \idmat{d}, \quad D(x) = \sqrt{1 - \rho' \rho(x)}\idmat{d},\\
 & r(x) = r_0+ \trace{r_1 x},  \quad \sigma(x) = \zeta(x) \sqrt{x}, \quad \mu(x) = \zeta(x)x\zeta'(x)\nu(x); \quad \text{ for } x\in \sdpos,
\end{align*}
where $r_0\in \Real$ and $r_1 \in \mathbb{M}^d$. We assume that $\nu\in C^{1,\gamma}(\sdpos; \Real^n)$, $\zeta \in C^{2,\gamma}(\sdpos; \mathbb{M}^{n\times d})$, and $\rho \in C^{2,\gamma}(\sdpos; \Real^d)$ are all bounded functions and $\sup_{x\in \sdpos} \rho'\rho(x) < 1$. When these functions are not constant, the previous model is not affine, in contrast to \cite{buraschi2010correlation, Hata-Sekine, Bauerle-Li, Richter}. For the given $\sigma$, Assumption \ref{ass: sig_ellip} takes the form
\begin{ass}\label{ass: sig_ellip_gW}\text{}
\begin{enumerate}[i)]
\item When $d> n$, $\zeta\zeta'(x) > 0$ for $x\in\sdpos$.
\item When $d < n$, $\zeta'\zeta(x) > 0$ for $x\in\sdpos$.
\item When $d=n$, $\zeta(x) = \zeta'(x) > 0$ for $x\in\sdpos$.
\end{enumerate}
\end{ass}

The following proposition verifies Statements \ref{stat: long hor} and \ref{stat: turnpike} in the current model under explicit parameter restrictions. The proof of Proposition \ref{prop: wishart} is in Appendix \ref{app: B}.

\begin{prop}\label{prop: wishart}
 Let  Assumption \ref{ass: sig_ellip_gW} hold. Assume the following parameter restrictions:
 \begin{enumerate}[i)]
  \item $LL' > (d+1)\Lambda \Lambda'>0$.
  \item When $p<0$, $r_1$ satisfies $r_1 + r_1 \geq 0$ and there exists $\epsilon > 0$ such that either
\begin{equation*}
-p(r_1+r_1')+q \zeta' \nu\nu'\zeta(x)\geq \epsilon \,\idmat{d},\qquad x\in\sdpos;
\end{equation*}
or
\begin{equation*}
(K-q\Lambda \rho \nu' \zeta)(x) + (K-q\Lambda \rho \nu' \zeta)'(x) \leq -\epsilon \,\idmat{d},\qquad x\in \sdpos.
\end{equation*}
  \item When $0<p<1$, there exists $\epsilon>0$ such that
\begin{equation*}
(K-q \Lambda \rho \nu' \zeta)(x) + (K- q \Lambda \rho \nu' \zeta)'(x) \leq -\epsilon \idmat{d},\qquad x\in \sdpos;
\end{equation*}
and
\begin{equation}\label{eq: p>0 cond}
\epsilon^2 > 8(1-q) \sqrt{d} \,\trace{\Lambda \Lambda'} \sup_{x\in \sdpos}\norm{p(r_1 + r_1') - q\zeta' \nu \nu' \zeta(x)}.
\end{equation}
\end{enumerate}
Then, the long-horizon convergence results in Statement \ref{stat: long hor} hold. Additionally when $r_1 = 0$, the turnpike theorems in Statement \ref{stat: turnpike} hold for all utility functions $U$ satisfying Assumption \ref{ass: ratio}.
\end{prop}

In the previous parameter restrictions, part i) is slightly stronger than the well-posedness condition \eqref{eq: wishart_cond}. The restriction in the $p<0$ case is \emph{mild}.
When  $r_1 + r_1'>0$, it follows that $-p(r_1 + r_1') + q \sigma' \nu \nu' \sigma (x) \geq \epsilon \, \idmat{d}$ for some $\epsilon>0$ since $q\zeta'\nu\nu'\zeta \geq 0$. Thus, part $ii)$ holds. When $r_1+r_1'$ is non-negative but may degenerate, consider a (generalized) Wishart process $\overline{X}$ with dynamics
\[
 d\overline{X}_t = \pare{L L' + \overline{K}(\overline{X}_t) \overline{X_t} + X_t \overline{K}(\overline{X}_t)}' dt + \sqrt{\overline{X}_t} dB_t \Lambda' + \Lambda
dB'_t \sqrt{\overline{X}_t},
\]
where $\overline{K}(x):=K- q \Lambda \rho \nu' \zeta(x)$.\footnote{This SDE admits a unique global strong solution $\overline{X}$. This is because $H_0(x; b)\geq 2 \trace{\overline{K}(x)}$ which is uniformly bounded from below due to the boundedness assumption of $\rho, \nu$, and $\zeta$ on $\sdpos$. Hence the existence follows from \cite[Theorem 3.4]{Mayerhofer-Pfaffel-Stelzer}.} Then  we require $\overline{X}$ is mean-reverting to verify part $ii)$. When $0<p<1$, we require the force of mean-reversion to be sufficiently strong. In this case, \eqref{eq: p>0 cond} is necessary because the potential
\begin{equation*}
V(x) = pr_0 + p\trace{r_1 x} - \frac{1}{2}q\nu'\zeta(x) x \zeta(x)'\nu = pr_0 +\frac{1}{2}\trace{x(p(r_1+r_1') - q\zeta'\nu\nu'\zeta(x))},
\end{equation*}
may not be uniformly bounded from above on $\sdpos$.

\subsubsection{An Explicit Long Run Optimal Strategy and a Counter-Example}\label{subsubsec: wishart_examples} We now focus on the ``classical'' Wishart model where $\rho,\nu$ and $\zeta$ in the previous section are constants taking values in $\Real^d$, $\Real^n$ and $\mdim{n}{d}$ respectively.  Here, it is shown that if the dimension $d$ of the Wishart process is \emph{less than or equal to} $n$, the number of risky assets, then the solution $\hat{v}$ to \eqref{eq: v_ergodic} with minimal $\hat{\lambda}$ is an affine function of $x$: i.e. up to an additive constant, $\hat{v}(x) = \textrm{Tr}(\hat{M}x)$ for a symmetric matrix $\hat{M}$ satisfying the Riccati equation given in \eqref{eq: d_leq_n_Ricatti} below. However, surprisingly, if $d > n$ then $\hat{v}$ may \emph{not} be affine, hence $\hat{\pi}$ in \eqref{eq: opt_strat} is not affine either. This is due to the non-commutative property of matrix product.

To streamline the presentation, we assume that $p<0$ and  $r_1 + r_1' > 0$. Hence Proposition \ref{prop: wishart} follows if $LL' > (d+1)\Lambda\Lambda' > 0$ and the constant matrix $\zeta$ satisfies Assumption \ref{ass: sig_ellip_gW}.  We consider candidate solutions to \eqref{eq: v_ergodic} given by
\begin{equation}\label{eq: cand_hatv}
v(x) = \trace{Mx},\qquad M=M'.\footnote{We can assume $M=M'$ without loss of generality since $x\in \sdpos$  implies  $\trace{Mx} = \trace{M'x} = (1/2)\trace{(M+M')x}$.}
\end{equation}
First we  present the result when $d\leq n$:

\begin{prop}\label{prop: wishart_good_case}
Assume $d\leq n$ and $\rho,\nu,\zeta$ are constant. Let $\zeta$ satisfy Assumption \ref{ass: sig_ellip_gW} and assume $p<0$, $r_1+r_1' > 0$, $LL' > (d+1)\Lambda\Lambda' >0$. Consider the following matrix Riccati equation in $M$:
\begin{equation}\label{eq: d_leq_n_Ricatti}
0 = 2M\Lambda(1-q\rho\rho')\Lambda'M + (K-q\Lambda\rho\nu'\zeta)'M + M(K-q\Lambda\rho\nu'\zeta) + \frac{1}{2}\left(p(r_1+r_1') - q\zeta'\nu\nu'\zeta\right).
\end{equation}
There exists a unique $\hat{M}\in\sd$ solving \eqref{eq: d_leq_n_Ricatti} such that $(\hat{\lambda}, \hat{v})$, with $\hat{\lambda}= \textrm{Tr}(LL'\hat{M})+ p r_0$ and $\hat{v}(x)=\textrm{Tr}(\hat{M}x)$, solves \eqref{eq: v_ergodic} and $\hat{\lambda}$ is the smallest $\lambda$ with accompanying $v$.
\end{prop}

We next present a counter-example in the $d > n$ case showing that solutions $(\hat{\lambda},\hat{v})$ to \eqref{eq: v_ergodic} \emph{cannot} be of the affine form in \eqref{eq: cand_hatv}. However the existence of solutions to \eqref{eq: v_ergodic} is still ensured by Proposition \ref{prop: wishart}.

\begin{exa}\label{exa: counter-exa}
Take $n=1, d=2$ and
\begin{equation}\label{eq: example_coeffs}
\begin{split}
\Lambda &= \idmat{2},\quad L = \ell \idmat{2}\textrm{ for }\ell > \sqrt{3}, \quad K = \idmat{2}, \quad C= \idmat{2},\\
\zeta &= \left(\begin{array}{c c} 1 & 0 \end{array}\right),\quad \nu = \nu\in\Real,\quad \rho = \rho\left(\begin{array}{c c} 1 &1\end{array}\right)'\textrm{ for } 0<2\rho^2 < 1,\\
r_0 &> 0, \quad r_1 = r_1\idmat{2}\textrm{ for } r_1 > 0.
\end{split}
\end{equation}
Consider functions $v$ as in \eqref{eq: cand_hatv}. Writing the generic element $X\in \sdpos$ and the matrix $M$ as
\begin{equation}\label{eq: wishart_x_v_ex}
X= \left(\begin{array}{c c} x & y \\ y & z\end{array}\right),\quad x,z > 0, y^2 < xz,\qquad M = \left(\begin{array}{c c} M_1 & M_2 \\ M_2 & M_3\end{array}\right),
\end{equation}
we have that $\Sigma(X) = \zeta X \zeta' = x>0$ so that Assumption \ref{ass: sig_ellip_gW} holds. Furthermore, $LL' - 3\Lambda\Lambda' = (\ell^2 - 3)\idmat{2} > 0$ and for $p<0$, $-p(r_1+r_1') + q\zeta'\nu\nu'\zeta (x) \geq -2pr_1 \idmat{2} > 0$.  Thus, the assumptions of Proposition \ref{prop: wishart} hold for $p<0$. A lengthy calculation shows that (cf. Lemma \ref{lem: example_calc} in Appendix \ref{app: C})
\begin{equation}\label{eq: F_example}
\begin{split}
\mathfrak{F}[v] &= x\left(2(M_1^2 + M_2^2) - 2q\rho^2(M_1+M_2)^2 + 2M_1 - 2q\rho\nu(M_1+M_2) + pr_1 - \frac{1}{2}q\nu^2\right)\\
&\qquad + y\left(4M_2(M_2+M_3) - 4q\rho^2(M_1+M_2)(M_2+M_3) + 4M_2 - 2q\rho\nu(M_2+M_3)\right)\\
&\qquad + z\left(2(M_2^2+M_3^2) +2M_3 + pr_1\right)\\
&\qquad + \frac{y^2}{x}\left(-2q\rho^2(M_2+M_3)^2\right)\\
&\qquad + pr_0 + \ell^2(M_1 + M_3).
\end{split}
\end{equation}
As can be seen from \eqref{eq: wishart_bar_A} in Lemma \ref{lem: wishart_op_affine_v} below, the problem term $y^2/x$ arises when evaluating $\bar{A}$ from \eqref{eq: functions}, since for $d>n$:
\begin{equation}\label{eq: wishart_counter_ex_Theta}
\sqrt{X}\Theta(X)\sqrt{X} =  X\zeta'(\zeta X\zeta')^{-1}\zeta X = \frac{1}{x}X\left(\begin{array}{c} 1\\ 0\end{array}\right)\left(\begin{array}{c c} 1 & 0\end{array}\right) X = \left(\begin{array}{c c} x & y \\ y & \frac{y^2}{x}\end{array}\right);
\end{equation}
whereas, for arbitrary model coefficients, if $d\leq n$ then $\sqrt{X}\Theta(X)\sqrt{X} = X$.

Thus, if $\mathfrak{F}[v] = \lambda$ for some constant $\lambda$ it must be that each coefficient of $x,y,z, y^2/x$ in \eqref{eq: F_example} is equal to zero. By considering $y^2/x$ it follows that $M_2 + M_3 = 0$. Plugging this into the coefficient of $y$
gives $M_2 = 0$ and hence $M_3=0$.  Then the coefficient of $z$ being zero yields $0 = pr_1$ a contradiction since $r_1 > 0$. Thus, the function $\hat{v}$ cannot be affine.
\end{exa}

\subsection{General State Variables}\label{subsec: general_case}

We now consider the general case when $X$ has dynamics as in \eqref{eq: state} where, in addition to the aforementioned regularity restrictions, the model coefficients satisfy Assumptions \ref{ass: wp} and \ref{ass: loc_ellip}.  As in the previous section, the goal is to provide conditions, based entirely upon the model coefficients, under which Statements \ref{stat: long hor} and \ref{stat: turnpike} hold.

To list the coefficient assumptions, let $f,g$ be as in \eqref{eq: f_g_def}, $\bar{b},V$ as in \eqref{eq: x_ops}, and recall $H_\delta(x;b)$ from \eqref{eq: Hdelta_def}. Assumption \ref{ass: coeff_master_list} below gives a number of restrictions under which the main convergence results hold.  Though the list below is lengthy, it can be readily checked for particular models of interest.

\begin{ass}\label{ass: coeff_master_list}
There exists $n_0>0$ such that the following hold for $\norm{x}\geq n_0$:
\begin{enumerate}[1)]
\item $\bar{b}$ has at most linear growth.
\item There exists $\alpha_1 > 0$ so that $\trace{f(x)}\trace{g(x)} \leq
  \alpha_1\norm{x}$.
\item There exits $\beta_1\in\Real$, $C_1 > 0$ so that $\trace{\bar{b}(x)'x} \leq
  -\beta_1\norm{x}^2 + C_1$.
\item There exists $\gamma_1,\gamma_2\in\Real$ and $C_2 > 0$ so that
  $-\gamma_2\norm{x} - C_2 \leq V(x) \leq -\gamma_1\norm{x} + C_2$. $V(x)$ is uniformly bounded from above for $\norm{x}\leq n_0$.
\item $\max\cbra{\gamma_1,\beta_1} > 0$.  Furthermore
\begin{enumerate}[i)]
\item If $\gamma_1 > 0, \beta_1\leq 0$, then there exist $\alpha_2>0$,
  $C_3\in\Real$ so that $\trace{f(x)xg(x)x} \geq \alpha_2\norm{x}^3 - C_3$.
\item If $\gamma_1< 0, \beta_1 > 0$, then
   $\beta_1^2 + 16\ok\alpha_1\gamma_1 > 0$, where  $\alpha_1$ is from part $2)$, $\ok=1$ when $p<0$, and $\ok=1-q$ when $0<p<1$.
\item If $\gamma_1 \geq 0, \beta_1 > 0$ then no additional restrictions are
  necessary.
\end{enumerate}
\end{enumerate}
There exists $\eps, c_0, c_1 > 0$ such that
\begin{enumerate}[A)]
\item $\inf_{x\in\sdpos} H_\eps(x; \bar{b})
  > -\infty$ (note : here we are using $\bar{b}$ instead of $b$ in \eqref{eq: Hdelta_def}).
\item $\liminf_{\det{x}\downarrow 0}\left(H_{\eps}(x; \bar{b}) + c_0
    \log(\det{x})\right) > -\infty$.
\item $\lim_{\det{x}\downarrow
    0}\left(H_0(x;\bar{b}) + c_1 V(x)\right) = \infty$.
\end{enumerate}
\end{ass}

\begin{rem} When $p<0$ and the interest rate function $r(x)$ is bounded from below on $\sdpos$ (e.g. $r(x)\geq 0$), then $\gamma_1 \geq 0$, hence the complicated part $5-ii)$ in Assumption \ref{ass: coeff_master_list} is never required.
\end{rem}

The parameter restrictions in Assumption \ref{ass: coeff_master_list} have a similar interpretation to those in Proposition \ref{prop: wishart}. Indeed, consider a $\sdpos$-valued diffusion $X$ with dynamics:
 \begin{equation}\label{eq: sde barX}
  d\bar{X}^{ij}_t = \bar{b}_{ij}(\bar{X}_t) dt + \trace{a^{ij}(\bar{X}_t) dW'_t}, \qquad i,j=1,\cdots, d.
 \end{equation}
 Comparing to \eqref{eq: state_alt}, the drift is adjusted to $\bar{b}$. The given regularity assumptions and parts 1) and 2) imply that the coefficients of $\bar{X}$ are locally Lipschitz and have at most linear growth. On the other hand, due to the second inequality in \eqref{eq: ellip}, $H_\delta$ is decreasing in $\delta$. Hence part A) implies $H_0(x; \bar{b})$ is bounded from below on $\sdpos$. As a result, Assumption \ref{ass: wp} specified to $\overline{X}$ from \eqref{eq: sde barX} holds and \cite[Theorem 3.4]{Mayerhofer-Pfaffel-Stelzer} ensures that \eqref{eq: sde barX} has a unique global strong solution.

 In Assumption \ref{ass: coeff_master_list} parts 3) and 4), if $\beta_1 >0$ then $\bar{X}$ is mean-reverting and if $\gamma_1>0$, the potential $V$ decays to $-\infty$ uniformly as $\norm{x}\rightarrow\infty$. Thus, part 5) requires either mean reversion or a decaying potential. If both happen, then no additional parameter restrictions is necessary. However, if mean reversion fails we require uniform ellipticity for $A(x)$ in the direction of $x$. If $\gamma_1<0$, then a delicate relationship in $5-ii)$ between the growth and degeneracy of $A$, mean reversion of $\bar{b}$ and the growth of $V$ is needed.

 Finally, Assumption \ref{ass: coeff_master_list} parts B) and C) are restrictions when the determinant of $\bar{X}$ is small. These two assumptions help to bound the value function $v$ from above and below, ensuring $v$ is finite close to the boundary $\{x\in \sdpos\,:\, \det(x) =0\}$ of the state space.

 From a technical point of view, Assumption \ref{ass: coeff_master_list} helps to construct an upper bound for solutions to \eqref{eq: v_pde}. It is shown in \cite[Section 3]{Robertson-Xing} that well-posedness of  \eqref{eq: v_pde} is established among solutions which are bounded from above (up to an additive constant) by
 \[
  \phi_0(x) := -\underline{c} \log(\det(x)) + \overline{c} \norm{x} \eta(\norm{x}) + C,
 \]
 where $\underline{c}, \overline{c} > 0$ and $C>0$ is chosen so that $\phi_0$ is non-negative on $\sdpos$. Here, $\eta\in C^{\infty}(0,\infty)$ is a cutoff function satisfying $0\leq \eta \leq 1$, $\eta(x) =1$ when $x>n_0+2$ and $\eta(x)=0$ for $x<n_0+1$, for the given $n_0$. Assumption \ref{ass: coeff_master_list} helps to verify the heuristic argument in Section \ref{subsec: inv prob}: \cite[Propositions 2.5, 2.7, and Theorem 3.9]{Robertson-Xing} prove that

 \begin{prop}\label{prop: v wellposed}
  Let Assumptions \ref{ass: loc_ellip}, \ref{ass: sig_ellip}, \ref{ass: rho} and \ref{ass: coeff_master_list} hold. Then there exists a unique solution $v\in C^{1,2}((0,\infty)\times \sdpos) \cap C([0,\infty)\times \sdpos)$ to \eqref{eq: v_pde} such that
  \[
   \sup_{(t,x)\in [0,T]\times \sdpos} (v(t,x)-\phi_0(x))<\infty, \quad \text{ for each } T\geq 0.
  \]
 \end{prop}

 Combining with the following verification result whose proof is deferred to Appendix \ref{app: A}, we obtain that the optimization problem in \eqref{eq: power op} is well-posed for any horizon $T>0$.

 \begin{prop}\label{prop: verification}
  Let Assumptions \ref{ass: loc_ellip}, \ref{ass: sig_ellip}, \ref{ass: rho} and \ref{ass: coeff_master_list} hold. Then for $v$ in Proposition \ref{prop: v wellposed} and any $T>0$, \eqref{eq: power op} holds and  $\pi^T$ from \eqref{eq: opt_strat_T} is the optimal strategy for \eqref{eq: power op}.
 \end{prop}

The aforementioned parameter assumptions also ensure the well-posedness of \eqref{eq: v_ergodic}: \cite[Proposition 2.3 and Lemma 5.3]{Robertson-Xing} prove that

\begin{prop}\label{prop: ergodic wellposed}
 Let Assumptions \ref{ass: loc_ellip}, \ref{ass: sig_ellip}, \ref{ass: rho} and \ref{ass: coeff_master_list} hold. There exists $(\hat{\lambda}, \hat{v})$ solving \eqref{eq: v_ergodic} such that $\hat{v}$ is unique (up to an additive constant) and  $\hat{\lambda}$ is  the smallest $\lambda$ such that there exists a corresponding  $v$ solving \eqref{eq: v_ergodic}.
\end{prop}

We are now ready to state our first main result, whose proof is presented in Appendix \ref{app: B}.

\begin{thm}\label{thm: power}
 Let Assumptions \ref{ass: loc_ellip}, \ref{ass: sig_ellip}, \ref{ass: rho} and \ref{ass: coeff_master_list} hold. Then the long horizon results in Statement \ref{stat: long hor} hold.
\end{thm}

To state the portfolio turnpike result, we need to make an additional assumption which is a mild strengthening of Assumption \ref{ass: rho}:

\begin{ass}\label{ass: rho_strong}
For $\rho$ and $C$ in Assumption \ref{ass: rho},  $\rho'\rho CC'(x) < \idmat{m}$ for all $x\in\sdpos$.
\end{ass}

Under the previous assumption, it is possible to construct not only super-martingale deflators (cf. \eqref{eq: M_eta_def} below), but also equivalent local martingale measures $\qprob^T$, for all $T>0$; i.e. $\qprob^T$ is equivalent to $\prob$ on $\F_T$ and $e^{-\int_0^\cdot r(X_u)\,du}S$ is a $\qprob^T$ local martingale on $[0,T]$. This is needed to utilize  duality results in \cite{karatzas.zitkovic.03} to establish the existence of an optimal strategy to \eqref{eq: op} for the generic utility $U$.

We are now ready to state the following turnpike result:

\begin{thm}\label{thm: turnpike}
Let Assumptions \ref{ass: wp}, \ref{ass: loc_ellip}, \ref{ass: sig_ellip}, \ref{ass: coeff_master_list} and \ref{ass: rho_strong} hold.  Then the turnpike theorems in Statement \ref{stat: turnpike} hold.
\end{thm}

\appendix

\section{Proof of Proposition \ref{prop: verification}}\label{app: A}

We first define a class of supermartingale deflators on $[0,T]$ for any $T>0$. Given a $\mathbb{M}^{d}$-valued process $\eta$ with $\int_0^T \norm{\eta_u}^2 du<\infty$ a.s., define $M^\eta$ via (note: for a function $g$ of $\sdpos$ we will write $g_u$ for $g(X_u)$):
\begin{equation}\label{eq: M_eta_def}
\begin{split}
M^{\eta}_t &\dfn
e^{-\int_0^t r_udu}\mathcal{E}\left(\int\left(-\nu_u'\sigma_u C_u dB_u \rho_u  + \trace{\eta_u dB_u'} - \rho_u'\eta_u'C_u'\Theta_uC_u dB_u\rho_u\right)\right)_t\\
&\qquad\times\mathcal{E}\left(-\int\left(\nu_u'\sigma_u D_u +\rho_u'\eta_u'C_u'\Theta_uD_u\right)dW_u\right)_t,\\
&= e^{-\int_0^t r_u du}\mathcal{E}\left(\int \sum_{k,l=1}^d dB^{kl}_u\left(-(C'\sigma'\nu)_k\rho_l + \eta_{kl} - (C'\Theta C\eta \rho)_k\rho_l\right)_u\right)_t\\
&\qquad\times\mathcal{E}\left(-\int \sum_{k=1}^d dW^k_u\left((D'\sigma'\nu)_k + (D'\Theta C\eta\rho)_k\right)_u\right)_t,\qquad t\leq T.
\end{split}
\end{equation}
When $\eta=0$, $e^{\int_0^\cdot r_u du} M^\eta$ defines the \emph{minimal martingale measure}, provided the stochastic exponentials are indeed martingales, see \cite{MR1108430}. Hence we call $\eta$ a \emph{risk premia}. For any admissible strategy $\pi$, $M^\eta \W^\pi$ is a positive super-martingale. Indeed, using \eqref{eq: BM Z}, \eqref{eq: wealth_dyn}, and \eqref{eq: M_eta_def}, the stochastic integration by parts formula shows that the drift of $M^\eta \W^\pi$ has the following integrand (omitting function arguments and time subscripts):
\begin{equation*}
\begin{split}
M^{\eta}\W^{\pi}&\pi' \bra{\Sigma\nu + \sigma C\left(-C'\sigma'\nu\rho' + \eta - C'\Theta C\eta\rho\rho'\right)\rho - \sigma D\left(D'\sigma'\nu + D'\Theta C\eta\rho\right)}\\
= &M^{\eta}\W^{\pi}\pi'\bra{\Sigma\nu - \sigma\left(CC'\rho'\rho + DD'\right)\sigma'\nu + \sigma C \eta\rho - \sigma\left(CC'\rho'\rho + DD'\right)\Theta C\eta\rho},\\
= &M^{\eta}\W^{\pi}\pi'\bra{\sigma C\eta\rho - \sigma\Theta C\eta\rho},\\
= &0,
\end{split}
\end{equation*}
where the second identity follows from $(CC'\rho'\rho + DD')(x)= 1_m$ and the third identity holds due to $\sigma\Theta = \sigma$. Therefore $M^\eta \W^\pi$ is a positive local martingale hence a super-martingale.

Before proving Proposition \ref{prop: verification}, we must introduce some notation. For a fixed $\phi\in C^{(1,2),\gamma}((0,\infty)\times \sdpos, \Real)$, the regularity assumptions on the coefficients and ellipticity assumption in \eqref{eq: ellip} ensure that the \emph{generalized} martingale problem on $\sdpos$ for
\begin{equation}\label{eq: cL_phi_time}
 \cL^{\phi,T-t} := \frac12 \sum_{i,j,k,l=1}^d A_{(ij),(kl)} D_{(ij),(kl)} + \sum_{i,j=1}^d \pare{\bar{b}_{ij} + \sum_{k,l=1}^d \bar{A}_{(ij),(kl)} D_{(kl)} \phi(T-t,\cdot)} D_{(ij)}, \quad t\leq T,
\end{equation}
has a unique solution $\left(\prob^{\phi,T,x}\right)_{x\in\sdpos}$ cf. \cite{Pinsky}.  When $\phi$ does not depend upon $t$ we will write $\cL^{\phi}$ and denote the solution as $\left(\prob^{\phi,x}\right)_{x\in\sdpos}$. The martingale problem for $\cL^{\phi,T-\cdot}$ is \emph{well-posed} if the coordinate process $X$ does not hit the boundary $\sdpos$, $\prob^{\phi,T,x}$-a.s., before $T$ for any $x\in\sdpos$. Similarly, if $\phi$ does not depend upon time, then well-posedness follows if the coordinate process does not hit the boundary in finite time $\prob^{\phi, x}$-a.s. for any $x\in\sdpos$.

For the given $\phi$, define the stochastic exponential
\begin{equation}\label{eq: Z_phi}
\begin{split}
Z^{\phi,T}_t \dfn &\mathcal{E}\left(\int_0^\cdot \sum_{k,l=1}^d dB^{kl}_u \left(-q(C'\sigma'\nu)_k\rho_l + \sum_{i,j=1}^d \left(a^{ij}_{kl}  - q(C'\Theta Ca^{ij}\rho)_k\rho_l\right)D_{(ij)}\phi\right)(T-u,X_u)\right)_t\\
&\times \mathcal{E}\left(\int_0^\cdot \sum_{k=1}^m dW^k_u\left(-q(D'\sigma'\nu)_k - q\sum_{i,j=1}^d(D'\Theta Ca^{ij}\rho)_kD_{(ij)}\phi\right)(T-u,X_u)\right)_t,\qquad t\leq T.
\end{split}
\end{equation}
For $\phi$ not depending upon time, write $Z^\phi$ for $Z^{\phi,T}$ and note that $Z^{\phi}$ is defined for all $t\geq 0$.
Recall from Section \ref{subsec: state var} that Assumption \ref{ass: wp} ensures the well-posedness of \eqref{eq: state}. Hence the martingale problem for $L$ in \eqref{eq: state_gen} is well-posed. Now if the martingale problem for $\cL^{\phi,T-\cdot}$ is also well-posed, it follows from (\cite[Remark 2.6]{Cheridito-Filipovic-Yor}) that the first stochastic exponential on the right hand side of \eqref{eq: Z_phi} is a $\prob^x$-martingale on $[0,T]$. On the other hand, since $X$ and $W$ are $\prob^x$-independent, it follows from \cite[Lemma 4.8]{Karatzas-Kardaras} that $Z^{\phi, T}$ is also a $\prob^x$-martingale on $[0,T]$. Therefore, we may define a new measure $\prob^{\phi, T, x}$ on $\F_T$ via
$ d\prob^{\phi,T,x}/d\prob^x |_{\F_T} = Z^{\phi,T}_T$.
Moreover, Girsanov's theorem yields that $X$ has generator $\cL^{\phi,T-\cdot}$ under $\prob^{\phi,T,x}$. When $\phi$ does not have time argument and the martingale problem for $\cL^\phi$ is well-posed, the same argument as above yields that $Z^\phi$ is a $\prob^x$-martingale on $[0,\infty)$. Hence a new measure $\prob^{\phi, x}$ is defined via
$ d\prob^{\phi,x}/d\prob^x |_{\F_T} = Z^{\phi}_T$, $T\geq 0$.
Note that $\prob^{\phi,x}$ is consistently defined on $\vee_{T\geq 0} \F_T$. Lastly we recall that $\prob^{\phi}$ is \emph{ergodic} if $X$  is recurrent under $\prob^\phi$ and there exists an invariant probability measure.

\begin{rem}\label{rem: numeraire}
 Set $\phi=\hat{v}$ from Proposition \ref{prop: ergodic wellposed}, if $\prob^{\hat{v},x}$ is well defined, then Girsanov's theorem together with \eqref{eq: sde S} and \eqref{eq: Z_phi} yield the following dynamics of $S$ under $\prob^{\hat{v},x}$:
 \[
  \frac{dS^i_t}{S^i_t} = \pare{r(X_t) + \frac{1}{1-p} \pare{\Sigma \nu + \sum_{k,l=1}^d \sigma C a^{kl} \rho D_{(kl)} \hat{v}}(T-t, X_t)} dt + \sum_{j=1}^m \sigma_{ij}(X_t) d\hat{Z}^j_t, \quad i=1, \dots, n,
 \]
 where $\hat{Z}$ is a $\prob^{\hat{v},x}$ Brownian motion. Comparing the previous dynamics with $\hat{\pi}$ in \eqref{eq: opt_strat}, it follows that $\hat{\pi}$ is the optimal strategy for a logarithmic investor under $\prob^{\hat{v},x}$. Hence its associated wealth process $\hat{\W}$ has the \emph{num\'{e}raire} property, i.e., $\W/\hat{\W}$ is a $\prob^{\hat{v},x}$-supermartingale for any admissible wealth process $\W$.
\end{rem}

For the proof of Proposition \ref{prop: verification}, we prepare following two lemmas, whose proofs are postponed until after the proof of Proposition \ref{prop: verification}.

\begin{lem}\label{lem: barA}
 Let Assumptions \ref{ass: loc_ellip}, \ref{ass: sig_ellip} and \ref{ass: rho} hold.  Let $A$ and $\bar{A}$ be as in \eqref{eq: functions}. Set
 \begin{equation}\label{eq: kappa}
  \uk = \left\{\begin{array}{ll}1, & 0<p<1\\ 1-q, & p<0\end{array}\right. \quad \text{ and } \quad \ok = \left\{\begin{array}{ll}1-q, & 0<p<1\\ 1,& p<0\end{array}\right..
 \end{equation}
  Then, for all $x\in \sdpos$ and $\theta\in\sd$:
\begin{equation}\label{eq: A_barA_compare}
   \uk \sum_{i,j,k,l=1}^d \theta_{ij} A_{(ij),(kl)}(x) \theta_{kl} \leq \sum_{i,j,k,l=1}^d \theta_{ij}\overline{A}_{(ij),(kl)}(x)\theta_{kl} \leq \ok \sum_{i,j,k,l=1}^d \theta_{ij} A_{(ij),(kl)}(x)\theta_{kl}.
\end{equation}
\end{lem}

For $\eta\in C^{(1,2), \gamma}((0,\infty)\times \sdpos, \Real)$, define function $\eta: \sdpos \rightarrow \mathbb{M}^d$ via
\begin{equation}\label{eq: eta_v_map}
\eta_{kl}(t,x; \phi) \dfn \left(\sum_{i,j=1}^d a^{ij}_{kl}D_{(ij)}\phi\right)(t,x),\qquad k,l = 1,...,d,\, t\geq 0,x\in\sdpos.
\end{equation}
Define $\eta^T_t:= \eta(T-t, X_t; \phi)$, $t\in[0,T]$. When $\phi$ is $v$ from Proposition \ref{prop: v wellposed} (resp. $\hat{v}$ from Proposition \ref{prop: ergodic wellposed}), then $\eta(T-\cdot, X_\cdot; v)$ (resp. $\eta(X_\cdot; \hat{v})$) is expected to be the optimal risk premium for the dual problem of \eqref{eq: power op} (resp. its long run analogue).
The following result is the key to prove Proposition \ref{prop: verification}.

\begin{lem}\label{lem: identity}
 Let $\phi \in C^{(1,2), \gamma}((0,\infty) \times \sdpos, \Real)$ satisfy $\phi_t = \fF[\phi]$ on $(0,\infty)\times \sdpos$ where $\fF$ is defined in \eqref{eq: x_ops}. For any $T\geq 0$, let $\pi_t= \pi(T-t, X_t; \phi)$, $\eta_t = \eta(T-t, X_t; \phi)$, for $t\in[0,T]$, and let $\W^\pi$ and $M^\eta$ be the associated wealth process and super-martingale deflator respectively. Then, the following identities hold:
 \begin{equation}\label{eq: prop_1_1}
\begin{split}
p\log\left(\W^\pi_T\right) - p\log\left(\W^\pi_t\right) + \phi(0,X_T) - \phi(T-t,X_t) &=
\log\pare{Z^{\phi,T}_T} - \log\pare{Z^{\phi,T}_t},\\
q\log\left(M^\eta_T\right) - q\log\left(M^\eta_t\right) + (1-q)(\phi(0,X_T) -
\phi(T-t,X_t)) &=\log\pare{Z^{\phi,T}_T} - \log\pare{Z^{\phi,T}_t},
\end{split}
\end{equation}
where $Z^{\phi, T}$ is given in \eqref{eq: Z_phi}.
\end{lem}

Using Lemmas \ref{lem: barA} and \ref{lem: identity}, the proof of Proposition \ref{prop: verification} is now given.

\begin{proof}[Proof of Proposition \ref{prop: verification}]
Note that in \eqref{eq: kappa}, $0<\uk <\ok$ holds for both $0<p<1$ and $p<0$. Thus, \cite[Assumption 3.4]{Robertson-Xing} is ensured by Assumption \ref{ass: loc_ellip} and Lemma \ref{lem: barA}. Additionally,  \cite[Assumptions 3.5 and 3.6]{Robertson-Xing} are exactly Assumption \ref{ass: coeff_master_list} here. As the assumptions of \cite[Lemma 4.1]{Robertson-Xing} are verified, the well-posedness of the martingale problem for $\cL^{v, T-\cdot}$ follows from \cite[Lemma 4.1]{Robertson-Xing}. Since the martingale problem for $L$ is also well-posed, it then follows from the discussion after \eqref{eq: Z_phi} that $Z^{v,T}$ is a $\prob^x$-martingale. Applying Lemma \ref{lem: identity} to $v$, it then follows from \eqref{eq: prop_1_1} and $v(0,x)=0$ that
\begin{equation}\label{eq: duality}
\expec\bra{\left.\left(\frac{\W^\pi_T}{\W^{\pi}_t}\right)^p\right| \F_t} =e^{v(T-t, X_t)}=\left(\expec\bra{\left.\left(\frac{M^{\eta}_T}{M^{\eta}_t}\right)^q\right| \F_t}\right)^{1/(1-q)}, \quad \text{ for all } t\leq T.
\end{equation}
Therefore the optimality of $\pi$ follows from \cite[Lemma 5]{Guasoni-Robertson} and \eqref{eq: power op} is verified in the previous identity.
\end{proof}

\begin{proof}[Proof of Lemma \ref{lem: barA}]
From \eqref{eq: x_ops}:
\begin{equation*}\label{eq: A_barA_compare_2}
\sum_{i,j,k,l=1}^d \theta_{ij}\overline{A}_{(ij),(kl)}(x)\theta_{kl} = \sum_{i,j,k,l=1}^d \theta_{ij} \trace{a^{ij}(a^{kl})'}(x)\theta_{kl} - q\sum_{i,j,k,l=1}^d \theta_{ij}\rho'(a^{ij})'C'\Theta C a^{kl}\rho \theta_{kl}.
\end{equation*}
Define the matrix $Y$ via $Y_{kl} \dfn \sum_{i,j=1}^d a^{ij}_{kl}\theta_{ij}$, for $k,l=1,...,d$. It then follows that
\begin{equation*}\label{eq: bar_A_to_A_compare_1}
\sum_{i,j,k,l=1}^d \theta_{ij}\rho'(a^{ij})'C'\Theta C a^{kl}\rho \theta_{kl} = \rho'Y'C'\Theta C Y \rho.
\end{equation*}
We claim that
\begin{equation}\label{eq: Y_mat_bounds}
0 \leq \rho'Y'C'\Theta C Y\rho \leq \trace{YY'}.
\end{equation}
Admitting this fact, and plugging back in for $Y$ yields
\begin{equation}\label{eq: bar_A_to_A_compare_5}
0 \leq \sum_{i,j,k,l=1}^d \theta_{ij}\rho'(a^{ij})'C'\Theta C a^{kl}\rho \theta_{kl} \leq \sum_{i,j,k,l=1}^d \theta_{ij} \trace{a^{ij}(a^{kl})'}(x)\theta_{kl}.
\end{equation}
If $p<0$ then $q>0$ and \eqref{eq: A_barA_compare} holds for $\uk = 1-q$ and $\ok = 1$. If $0<p<1$ then $q<0$ and hence \eqref{eq: A_barA_compare} holds for $\uk = 1$ and $\ok = 1-q$.

It remains to show \eqref{eq: Y_mat_bounds}. When $\rho(x) = 0_d$, the $d$-dimensional vector with all components $0$, it is clear that $\rho'Y'C'\Theta C Y \rho=0$ and \eqref{eq: Y_mat_bounds} holds. When $\rho(x)\neq 0_d$, it follows from $\Theta\geq 0$ that $\rho'Y'C'\Theta C Y \rho \geq 0$. On the other hand, since by construction $\Theta\leq 1$ (see \eqref{eq: Theta_def}),  we have
\begin{equation*}
\rho'Y'C'\Theta C Y \rho \leq \rho'Y'C'C Y\rho \leq \frac{1}{\rho'\rho} \rho'Y'Y\rho = \frac{1}{\rho'\rho}\trace{Y\rho \rho'Y'},
\end{equation*}
where the second inequality holds by Assumption \ref{ass: rho} and the fact
that $C'C$ and $CC'$ have the same eigenvalues. Note that the eigenvalues of
$(1/\rho'\rho)\rho\rho'$ are $1$ and $0$,
and that $\trace{N M N'} \leq \lambda^{+,M}\trace{NN'}$ for any $n\in \mathbb{M}^d$ and $M\in \sd$, where $\lambda^{+,M}$ is the maximal eigenvalue of $M$. Therefore, $(1/\rho'\rho)\trace{Y\rho\rho'Y} \leq \trace{YY'}$ and \eqref{eq: Y_mat_bounds} is confirmed, finishing the proof.

\end{proof}

\begin{proof}[Proof of Lemma \ref{lem: identity}]
The proof  is similar that of \cite[Lemma B.3]{Guasoni-Robertson-fundsep}. However, since herein we work with a semi-linear equation and a matrix valued state variable, the notational differences in the calculations are such that, for clarity, we will present a detailed proof.

First of all,  set
\begin{equation}\label{eq: A_B_def}
\begin{split}
\textbf{A} &:= p\log\left(\W^\pi_T\right) - p\log\left(\W^\pi_t\right) + \phi(0,X_T)
- \phi(T-t,X_t),\\
\textbf{B} &:= q\log\left(M^\eta_T\right) - q\log\left(M^\eta_t\right) + (1-q)(\phi(0,X_T) -
\phi(T-t,X_t)).
\end{split}
\end{equation}
The identities in \eqref{eq: prop_1_1} are verified in the following four steps.
\begin{enumerate}[1)]
\item Use the dynamics for $\W^\pi$ in \eqref{eq: wealth_dyn}, the definition of $M^{\eta}$ in \eqref{eq: M_eta_def}, and the
  definitions of $\pi$, $\eta$ in \eqref{eq: pi_v_map} and \eqref{eq: eta_v_map}  to write
\begin{equation}\label{eq: AB iden}
\begin{split}
\textbf{A} = \int_t^T \textbf{A1}_udu + \sum_{k,l=1}^d
\int_t^T\textbf{A2}^{kl}_udB^{kl}_u +
\sum_{k=1}^{m}\int_t^T\textbf{A3}^k_u dW^k_u,\\
\textbf{B} = \int_t^T \textbf{B1}_udu + \sum_{k,l=1}^d
\int_t^T\textbf{B2}^{kl}_udB^{kl}_u + \sum_{k=1}^{m}\int_t^T\textbf{B3}^k_u dW^k_u,\\
\end{split}
\end{equation}
where $\textbf{A1},\textbf{B1}:[0,T]\times\sdpos\rightarrow\Real$, $\textbf{A2},\textbf{B2}:[0,T]\times\sdpos\rightarrow\mathbb{M}^d$, and $\textbf{A3},\textbf{B3}:[0,T]\times\sdpos\rightarrow\Real^{m}$. These functions with time subscripts represent, for example, $\textbf{A1}_u = \textbf{A1}(T-u, X_u)$.
\item Add and subtract
\begin{equation}\label{eq: addsub}
\begin{split}
&\frac{1}{2}\sum_{k,l=1}^d \int_t^T
  \left(\textbf{A2}^{kl}_u\right)^2 du + \frac{1}{2}\sum_{k=1}^{m}\int_t^T
  \left(\textbf{A3}^k_u\right)^2du,\\
&\frac{1}{2}\sum_{k,l=1}^d \int_t^T
  \left(\textbf{B2}^{kl}_u\right)^2 du+ \frac{1}{2}
  \sum_{k=1}^{m}\int_t^T\left(\textbf{B3}^k_u\right)^2du,
\end{split}
\end{equation}
to the right-hand-side of \textbf{A} and \textbf{B}, respectively, to obtain
\begin{equation*}
\begin{split}
\textbf{A} &= \int_t^T \left(\textbf{A1}_u+
  \frac{1}{2}\sum_{k,l=1}^d\left(\textbf{A2}_u^{kl}\right)^2 +
  \frac{1}{2}\sum_{k=1}^{m}\left(\textbf{A3}_u^k\right)^2\right)du
+ \log(\cZ_T) - \log(\cZ_t),\\
\textbf{B} &= \int_t^T \left(\textbf{B1}_u+
  \frac{1}{2}\sum_{k,l=1}^d\left(\textbf{B2}_u^{kl}\right)^2 +
  \frac{1}{2}\sum_{k=1}^{m}\left(\textbf{B3}^k_u\right)^2\right)du
+ \log(\tilde{\cZ}_T) - \log(\tilde{\cZ}_t),\\
\end{split}
\end{equation*}
where
\begin{equation}\label{eq: cZ def}
\begin{split}
\cZ = \mathcal{E}\left(\int \sum_{k,l=1}^d
  \textbf{A2}_u^{kl}dB^{kl}_u + \int
  \sum_{k=1}^{m}\textbf{A3}^{k}_udW^l_u\right), \quad
\tilde{\cZ} = \mathcal{E}\left(\int \sum_{k,l=1}^d
  \textbf{B2}^{kl}_udB^{kl}_u + \int
  \sum_{k=1}^{m}\textbf{B3}^k_udW^k_u\right).
\end{split}
\end{equation}
\item Show that for $u\leq T$ and $x\in\sdpos$:
\begin{equation*}
\begin{split}
\left(\textbf{A1}+
  \frac{1}{2}\sum_{k,l=1}^d\left(\textbf{A2}^{kl}\right)^2 +
  \frac{1}{2}\sum_{k=1}^{m}\left(\textbf{A3}^k\right)^2\right)(T-u,x) =
  \left(-\phi_t+ \fF[\phi]\right)(T-u,x) = 0,\\
\left(\textbf{B1}+
  \frac{1}{2}\sum_{k,l=1}^d\left(\textbf{B2}^{kl}\right)^2 +
  \frac{1}{2}\sum_{k=1}^{m}\left(\textbf{B3}^k\right)^2\right)(T-u,x) =
  \left(-\phi_t+ \fF[\phi]\right)(T-u,x) = 0.\\
\end{split}
\end{equation*}
\item Show that $\cZ = \tilde{\cZ} = Z^{\phi,T}$.
\end{enumerate}
Combining the above four steps, \eqref{eq: prop_1_1} is then verified.

\begin{rem}
For notational ease the following conventions are used: 1) we will omit  $\int_t^T$ and the integrator $du$ from all integrals; 2) we will suppress the argument $(T-u,X_u)$ from
all functions; 3) we will also drop all time subscripts.  Thus, for
example, we will write
\begin{equation*}
f + g'dB\rho + h'dW =\int_t^T f(T-u,X_u)du + \int_t^T g(T-u,X_u)'dB_u\rho(X_u) + \int_t^T h(T-u,X_u)'dW_u.
\end{equation*}
\end{rem}

The first identity in \eqref{eq: prop_1_1} is now shown. Using $\rho'\rho CC' + DD' = \idmat{m}$ and the dynamics of $\W^\pi$ in \eqref{eq: wealth_dyn}, \ito's formula gives \eqref{eq: AB iden} where
\begin{equation}\label{eq: A def}
\begin{split}
\textbf{A1} &= pr + p\pi'\Sigma\nu - \frac{1}{2}p\pi'\Sigma\pi - \phi_t + L\phi,\\
\textbf{A2}^{kl} &= p(C'\sigma'\pi)_k\rho_l + \sum_{i,j=1}^d
a^{ij}_{kl}D_{(ij)}\phi,\\
\textbf{A3}^k &= p(D'\sigma'\pi)_k.
\end{split}
\end{equation}

While the second step follows from definitions of $Z$ and $\tilde{Z}$, we move onto the third step. For $u\leq T$ and $ x\in\sdpos$, it follows that
\begin{equation}\label{eq: prop_1_2}
\begin{split}
\textbf{A1} +& \frac{1}{2}\sum_{k,l=1}^d (\textbf{A2}^{kl})^2 +
\sum_{k=1}^{m}(\textbf{A3}^k)^2\\
=& pr + p\pi'\Sigma\nu - \frac{1}{2}p\pi'\Sigma\pi - \phi_t +L\phi + \frac{1}{2}p^2\pi'\sigma CC'\sigma'\pi \rho'\rho + p\pi'\left(\sum_{i,j1}^d \sigma C a^{ij}\rho D_{(ij)}\phi \right)\\
&  + \frac{1}{2}\sum_{i,j,k,l=1}^d D_{(ij)}\phi \trace{a^{ij}(a^{kl})'} D_{(kl)}\phi + \frac{1}{2}p^2\pi'\sigma DD'\sigma'\pi,\\
=&\frac{1}{2}p(p-1)\pi'\Sigma\pi
+ p\pi'\Sigma\nu + p\pi'\left(\sum_{ij=1}^d \sigma C a^{ij}\rho D_{(ij)}\phi \right)\\
&  +pr - \phi_t + L\phi + \frac{1}{2}\sum_{i,j,k,l=1}^d D_{(ij)}\phi\
\trace{a^{ij}(a^{kl})'} D_{(kl)}\phi.
\end{split}
\end{equation}
The terms above containing $\pi$ are
\begin{equation*}
\frac{1}{2}p(p-1)\pi'\Sigma\pi + p\pi'\left(\Sigma\nu + \sum_{i,j=1}^d
  \sigma C a^{ij}\rho D_{(ij)}\phi\right).
\end{equation*}
Using  \eqref{eq: pi_v_map}, we obtain the following expression for the quadratic function in the previous line:
\[
 -\frac{1}{2}q\nu'\Sigma\nu - q\sum_{i,j=1}^d\ \nu'\sigma C
a^{ij}\rho D_{(ij)}\phi - \frac{1}{2}q \sum_{i,j,k,l=1}^d D_{(ij)}\phi\ \rho'(a^{ij})'C'\Theta
C a^{kl}\rho D_{(kl)}\phi,
\]
for both cases $m\geq n$ or $m<n$. Thus, substituting the previous expression into \eqref{eq: prop_1_2}, using the expressions for $\bar{A}, V$ in \eqref{eq: functions} and $\fF$ in \eqref{eq: x_ops_L} gives
\begin{equation}\label{eq: prop_1_3}
\begin{split}
\textbf{A1} +& \frac{1}{2}\sum_{k,l=1}^d (\textbf{A2}^{kl})^2 +
\sum_{k=1}^{m}(\textbf{A3}^k)^2\\
=& pr-\frac{1}{2}q\nu'\Sigma\nu - q\sum_{i,j=1}^d \nu'\sigma C a^{ij}\rho D_{(ij)}\phi - \frac{1}{2}\sum_{i,j,k,l=1}^d D_{(ij)}\phi \rho'(a^{ij})'C'\Theta C a^{kl}\rho D_{(kl)}\phi \\
&- \phi_t + L\phi + \frac{1}{2}\sum_{i,j,k,l=1}^d D_{(ij)}\phi \trace{a^{ij}(a^{kl})'} D_{(kl)}\phi\\
=& -\phi_t + L\phi -
q\sum_{i,j=1}^d\nu'\sigma C a^{ij}\rho D_{(ij)}\phi + \frac{1}{2}\sum_{i,j,k,l=1}^d D_{(ij)}\phi \bar{A}_{(ij),(kl)} D_{(kl)}\phi + V\\
=& -\phi_t + \fF[\phi] \\
=&0,
\end{split}
\end{equation}
finishing the third step. For the last step, recall the definition of $Z^{\phi,T}$ from \eqref{eq: Z_phi}.
Comparing with the definition of $\cZ$ in \eqref{eq: cZ def}, it suffices to show that
\begin{equation}\label{eq: A2_A3}
\begin{split}
\textbf{A2}^{kl} &= -q(C'\sigma'\nu)_k\rho_l + \sum_{i,j=1}^d \left(
a^{ij}_{kl} - q(C'\Theta C a^{ij}\rho)_k\rho_l\right)D_{(ij)}\phi, \\
\textbf{A3}^k & = -q(D'\sigma'\nu)_k - q\sum_{i,j=1}^d
\left(D'\Theta C a^{ij}\rho\right)_kD_{(ij)}\phi.
\end{split}
\end{equation}
Using \eqref{eq: pi_v_map} for $m\geq n$ it follows that (recall $\Theta = \sigma'\Sigma^{-1}\sigma$ when $m\geq n$)
\begin{equation*}
\begin{split}
p(\sigma'\pi)&= -q\sigma'\Sigma^{-1}\left(\Sigma\nu + \sum_{i,j=1}^d \sigma C a^{ij}\rho D_{(ij)}\phi\right) =-q\sigma'\nu - q\sum_{i,j=1}^d \Theta C a^{ij}\rho D_{(ij)}\phi.
\end{split}
\end{equation*}
Similarly, using \eqref{eq: pi_v_map} for $m< n$ gives (recall $\Theta = 1_m$ for $m<n$):
\begin{equation*}
\begin{split}
p(\sigma'\pi)&= -q\sigma'\sigma(\sigma'\sigma)^{-1}\left(\sigma'\nu + \sum_{i,j=1}^d C a^{ij}\rho D_{(ij)}\phi\right)=-q\sigma'\nu - q\sum_{i,j=1}^d \Theta C a^{ij}\rho D_{(ij)}\phi.
\end{split}
\end{equation*}
Therefore, in both cases $m\geq n$, $m< n$ we have, using the definition of $\textbf{A2},\textbf{A3}$ in \eqref{eq: A def} that
\begin{equation*}
\begin{split}
\textbf{A2}^{kl} &= p(C'\sigma'\pi)_k \rho_l + \sum_{i,j=1}^d a^{ij}_{kl} D_{(ij)}\phi =-q(C'\sigma'\nu)_k\rho_l + \sum_{i,j=1}^d \left(a^{ij}_{kl} - q(C'\Theta C a^{ij}\rho)_k\rho_l\right) D_{(ij)}\phi,\\
\textbf{A3}^k &= p(D'\sigma'\pi)_k = -q(D'\sigma'\nu)_k - q\sum_{i,j=1}^d (D'\Theta C a^{ij}\rho)_k D_{(ij)}\phi,
\end{split}
\end{equation*}
which verifies \eqref{eq: A2_A3}.

The proof for the second identity in \eqref{eq: prop_1_1} is similar. First, using the definition of $M^\eta$ in \eqref{eq: M_eta_def}, \ito's formula yields the second identity in \eqref{eq: AB iden}, where
\begin{equation}\label{eq: B def}
\begin{split}
\textbf{B1} = & -qr + (1-q)(-\phi_t + L\phi)\\ &-\frac{1}{2}q\left(\sum_{k,l=1}^{d}\left(-(C'\sigma'\nu)_k\rho_l + \eta_{kl}
  - (C'\Theta C\eta\rho)_k\rho_l\right)^2 +
\sum_{k=1}^{m}\left((D'\sigma'\nu)_k + (D'\Theta
  C\eta\rho)_k\right)^2\right),\\
\textbf{B2}^{kl} &= q\left(-(C'\sigma'\nu)_k\rho_l + \eta_{kl} - (C'\Theta
C\eta\rho)_k\rho_l\right) + (1-q)\sum_{i,j=1}^d a^{ij}_{kl}D_{(ij)}\phi,\\
\textbf{B3}^{k} &= -q\left((D'\sigma'\nu)_k + (D'\Theta C\eta\rho)_k\right).
\end{split}
\end{equation}

Using $(1-q)p=-q$ we obtain
\begin{equation}\label{eq: prop_1_45}
\begin{split}
&\textbf{B1} + \frac{1}{2}\sum_{k,l=1}^d (\textbf{B2}^{kl})^2 + \frac{1}{2}\sum_{k=1}^{m} (\textbf{B3}^k)^2\\
= & (1-q)pr  + (1-q)(-\phi_t + L\phi)\\
& -\frac{1}{2}q(1-q)\left(\sum_{k,l=1}^{d}\left(-(C'\sigma'\nu)_k\rho_l + \eta_{kl}
  - (C'\Theta C\eta\rho)_k\rho_l\right)^2 +
\sum_{k=1}^{m}\left((D'\sigma'\nu)_k + (D'\Theta
  C\eta\rho)_k\right)^2\right)\\
&+ q(1-q)\sum_{i,j,k,l=1}^d \left(-(C'\sigma'\nu)_k\rho_l + \eta_{kl}
  - (C'\Theta C\eta\rho)_k\rho_l\right) a^{ij}_{kl}D_{(ij)}\phi\\
&+ \frac{1}{2}(1-q)^2\sum_{k,l=1}^d\left(\sum_{i,j=1}^da^{ij}_{kl}D_{(ij)}\phi\right)^2.
\end{split}
\end{equation}
Now, using $\rho'\rho CC' + DD' = \idmat{m}$ gives
\begin{equation*}
\begin{split}
 &\sum_{k,l=1}^{d}\left(-(C'\sigma'\nu)_k\rho_l + \eta_{kl}
  - (C'\Theta C\eta\rho)_k\rho_l\right)^2 +
\sum_{k=1}^{m}\left((D'\sigma'\nu)_k + (D'\Theta
  C\eta\rho)_k\right)^2\\
&=\nu'\sigma CC'\sigma' \nu \rho'\rho + \trace{\eta'\eta} + \rho'\eta'C'\Theta CC'\Theta C\eta\rho\rho'\rho - 2\nu'\sigma C\eta\rho + 2\nu'\sigma CC'\Theta C\eta\rho \rho'\rho -2\rho'\eta'C'\Theta C\eta \rho\\
&\qquad + \nu'\sigma DD'\sigma'\nu + \rho'\eta'C'\Theta DD'\Theta C\eta\rho + 2\nu'\sigma DD'\Theta C\eta\rho\\
&=\nu'\sigma(CC'\rho'\rho + DD')\sigma'\nu  + \rho'\eta'C'\Theta(CC'\rho'\rho + DD')\Theta C\eta\rho + 2\nu'\sigma(CC'\rho'\rho + DD')\Theta C\eta\rho\\
&\qquad +\trace{\eta'\eta} - 2\nu'\sigma C\eta\rho - 2\rho'\eta'C'\Theta C\eta\rho\\
&=\nu'\Sigma\nu  + \rho'\eta'C'\Theta\Theta C\eta\rho + 2\nu'\sigma\Theta C\eta\rho +\trace{\eta'\eta} - 2\nu'\sigma C\eta\rho - 2\rho'\eta'C'\Theta C\eta\rho\\
&=\nu'\Sigma\nu  + \trace{\eta'\eta} - \rho'\eta'C'\Theta C\eta\rho,
\end{split}
\end{equation*}
where the last equality follows since the definition of $\Theta$ in \eqref{eq: Theta_def} implies both $\Theta\Theta = \Theta$ and $\sigma\Theta = \sigma$.  We also have
\begin{equation*}
\begin{split}
&\sum_{i,j,k,l=1}^d \left(-(C'\sigma'\nu)_k\rho_l + \eta_{kl}
  - (C'\Theta C\eta\rho)_k\rho_l\right) a^{ij}_{kl}D_{(ij)}\phi \\
  &\hspace{2cm} = \sum_{i,j=1}^d \left(-\nu'\sigma C a^{ij}\rho + \trace{\eta' a^{ij}} - \rho'\eta'C'\Theta C a^{ij}\rho\right)D_{(ij)}\phi,\\
&\sum_{k,l=1}^d\left(\sum_{i,j=1}^da^{ij}_{kl}D_{(ij)}\phi\right)^2 = \sum_{i,j,k,l=1}^d D_{(ij)}\phi \trace{a^{ij}(a^{kl})'} D_{(kl)}\phi.
\end{split}
\end{equation*}
Plugging all of this into \eqref{eq: prop_1_45} yields
\begin{equation}\label{eq: prop_1_55}
\begin{split}
&\frac{1}{1-q}\pare{\textbf{B1} + \frac{1}{2}\sum_{k,l=1}^d (\textbf{B2}^{kl})^2 + \frac{1}{2}\sum_{k=1}^{m} (\textbf{B3}^k)^2}\\
=& pr - \phi_t + L\phi - \frac{1}{2}q\left(\nu'\Sigma\nu + \trace{\eta'\eta} - \rho'\eta'C'\Theta C\eta\rho\right) \\
&+ q\sum_{i,j=1}^d\left(-\nu'\sigma C a^{ij}\rho + \trace{\eta'a^{ij}} - \rho'\eta'C\Theta C a^{ij}\rho\right)D_{(ij)}\phi\\
&+ \frac{1}{2}(1-q)\sum_{i,j,k,l=1}^d D_{(ij)}\phi\trace{a^{ij}(a^{kl})'}D_{(kl)}\phi.
\end{split}
\end{equation}
On the right-hand-side, terms involving $\eta$ are
\begin{equation}\label{eq: prop_1_8}
\begin{split}
-\frac{1}{2}q\trace{\eta'\eta}& + \frac{1}{2}q\rho'\eta'C'\Theta C\eta\rho +
q\sum_{i,j=1}^d \trace{\eta'a^{ij}}D_{(ij)}\phi - q\sum_{i,j=1}^d
  \rho'\eta' C'\Theta C a^{ij}\rho D_{(ij)}\phi.
\end{split}
\end{equation}
For $\eta$ in \eqref{eq: eta_v_map}, the following identities hold
\begin{equation*}
\begin{split}
&\trace{\eta'\eta}  = \sum_{i,j,k,l=1}^d D_{(ij)}\phi\ \trace{a^{ij}(a^{kl})'} D_{(kl)}\phi,\\
&\rho'\eta'C'\Theta C\eta\rho = \sum_{i,j,k,l=1}^d D_{(ij)}\phi\ \rho'(a^{ij})'C'\Theta C a^{kl}\rho D_{(kl)} \phi,\\
&\sum_{i,j=1}^d \trace{\eta'a^{ij}}D_{(ij)}\phi = \sum_{i,j,k,l=1}^d D_{(ij)}\phi\ \trace{a^{ij}(a^{kl})'}
D_{(kl)}\phi,\\
&\sum_{i,j=1}^d \rho'\eta' C'\Theta C a^{ij}\rho D_{(ij)}\phi = \sum_{i,j,k,l=1}^d D_{(ij)}\phi\ \rho'(a^{ij})'C'\Theta C a^{kl}\rho D_{(kl)} \phi.
\end{split}
\end{equation*}
Using above identities in \eqref{eq: prop_1_8}, we obtain the following expression for \eqref{eq: prop_1_8}:
\[
 \frac{1}{2}q\sum_{i,j,k,l=1}^dD_{(ij)}\phi\ \left(\trace{a^{ij}(a^{kl})'} -
  \rho'(a^{ij})'C'\Theta C a^{kl}\rho\right)D_{(kl)}\phi.
\]
Inserting this into \eqref{eq: prop_1_55} gives
\begin{equation*}\label{eq: prop_1_95}
\begin{split}
&\frac{1}{1-q}\pare{\textbf{B1} + \frac{1}{2}\sum_{i,j=1}^d (\textbf{B2}^{ij})^2 + \frac{1}{2}\sum_{l=1}^{m} (\textbf{B3}^l)^2}\\
=& pr - \phi_t + L\phi - \frac{1}{2}q\nu'\Sigma\nu - q\sum_{i,j=1}^d \nu'\sigma C a^{ij}\rho D_{(ij)}\phi + \frac{1}{2}(1-q)\sum_{i,j,k,l=1}^d D_{(ij)}\phi \trace{a^{ij}(a^{kl})'} D_{(kl)}\phi\\
&+ \frac{1}{2}q\sum_{i,j,k,l=1}^d D_{(ij)}\phi\left(\trace{a^{ij}(a^{kl})'} - \rho'(a^{ij})'C'\Theta C a^{kl}\rho\right)D_{(kl)}\phi\\
=&-\phi_t + L\phi - q\sum_{i,j=1}^d \nu'\sigma C a^{ij}\rho D_{(ij)}\phi + \frac{1}{2}\sum_{i,j,k,l=1}^d D_{(ij)}\phi\left(\trace{a^{ij}(a^{kl})'} - q\rho'(a^{ij})'C'\Theta C a^{kl}\rho\right)D_{(kl)}\phi\\
& + pr - \frac{1}{2}q\nu'\Sigma\nu\\
=& -\phi_t + \fF[\phi]\\
=&0,
\end{split}
\end{equation*}
where the second to last equality uses \eqref{eq: functions} and \eqref{eq: x_ops_L}. Thus, the third step is complete.

Turning to the last step, comparing $Z^{\phi,T}$ in \eqref{eq: Z_phi} with $\tilde{\cZ}$ in \eqref{eq: cZ def}, it suffices to show
\begin{equation*}
\begin{split}
\textbf{B2}^{kl} &= -q(C'\sigma'\nu)_k\rho_l + \sum_{i,j=1}^d\left(
a^{ij}_{kl} - q(C'\Theta C a^{ij}\rho)_k\rho_l\right)D_{(ij)}\phi,\\
\textbf{B3}^k & = -q(D'\sigma'\nu)_k - q\sum_{i,j=1}^d\left(D'\Theta C a^{ij}\rho\right)_k D_{(ij)}\phi.
\end{split}
\end{equation*}
Using the definitions of $\textbf{B2}$ and $\textbf{B3}$ in \eqref{eq: B def} it suffices to show that
\begin{equation*}
\begin{split}
q\eta_{kl} - q(C'\Theta C\eta\rho)_k\rho_l + (1-q)\sum_{i,j=1}^d a^{ij}_{kl} D_{(ij)}\phi &=\sum_{i,j=1}^d \left(a^{ij}_{kl} - q(C'\Theta C a^{ij}\rho)_k\rho_l\right)D_{(ij)}\phi,\\
(D'\Theta C \eta\rho)_k &=\sum_{i,j=1}^d (D'\Theta C a^{ij}\rho)_k D_{(ij)}\phi.
\end{split}
\end{equation*}
Since $\eta_{kl} = \sum_{i,j=1}^d a^{ij}_{kl}D_{(ij)}\phi$ from \eqref{eq: eta_v_map} the last two identities readily follow, finishing the proof.

\end{proof}

\section{Proofs for Subsection \ref{subsubsec: wishart_examples}}\label{app: C}

Throughout this section, the model is from Section \ref{subsec: wishart} with $\rho,\nu$ and $\zeta$ constant.  Furthermore, $\zeta$ is assumed to satisfy Assumption \ref{ass: sig_ellip_gW}. We begin with the following lemma, which identifies $\mathfrak{F}[v]$ for $v$ as in \eqref{eq: cand_hatv}.

\begin{lem}\label{lem: wishart_op_affine_v} For $v = \trace{Mx}$ as in \eqref{eq: cand_hatv} it follows for $d\leq n$ that
\begin{equation}\label{eq: wishart_fv_affine_d_leq_n}
\begin{split}
&\mathfrak{F}[v](x) = \trace{x\left(2M\Lambda(1-q\rho\rho')\Lambda'M + K'M + MK - q\zeta'\nu\rho'\Lambda'M - qM\Lambda\rho\nu'\zeta + \frac{1}{2}\left(p(r_1+r_1') - q\zeta'\nu\nu'\zeta\right)\right)}\\
&\qquad\qquad + \trace{LL'M} + p r_0.
\end{split}
\end{equation}
For $d > n$
\begin{equation}\label{eq: wishart_fv_affine_d_ge_n}
\begin{split}
&\mathfrak{F}[v](x) = \trace{x\left(2M\Lambda\Lambda'M + K'M + MK - q\zeta'\nu\rho'\Lambda'M - qM\Lambda\rho\nu'\zeta + \frac{1}{2}\left(p(r_1+r_1') - q\zeta'\nu\nu'\zeta\right)\right)}\\
&\qquad - 2q\trace{x\zeta'\left(\zeta x\zeta'\right)^{-1}\zeta x M\Lambda\rho\rho'\Lambda'M} + \trace{LL'M} + p r_0.
\end{split}
\end{equation}

\end{lem}

\begin{proof}
Plugging in the model coefficients gives
\begin{equation*}
\begin{split}
b(x) &= LL' + Kx + xK',\qquad a^{ij}_{kl}(x) = \sqrt{x}_{ik}\Lambda_{jl} + \sqrt{x}_{jk}\Lambda_{il},\\
r(x) &=r_0 + \trace{r_1x},\qquad \sigma(x) = \zeta\sqrt{x},\qquad \nu(x) = \nu,\\
C(x) &=\idmat{d},\qquad \rho(x) = \rho.
\end{split}
\end{equation*}
Therefore, using the definitions in \eqref{eq: functions}, calculation shows that
\begin{equation}\label{eq: wishart_b_A_V}
\begin{split}
\bar{b}_{ij}(x) =& (LL'+Kx+xK')_{ij}-q(x\zeta'\nu\rho'\Lambda')_{ij} - q(x\zeta'\nu\rho'\Lambda')_{ji},\\
A_{(ij),(kl)}(x) =& x_{ik}(\Lambda\Lambda')_{jl} + x_{il}(\Lambda\Lambda')_{jk} +x_{jk}(\Lambda\Lambda')_{il} + x_{jl}(\Lambda\Lambda')_{ik},\\
V(x) =& pr_0 + \frac{1}{2}p\trace{x(r_1+r_1')} - \frac{1}{2}q\trace{x\zeta'\nu\nu'\zeta},
\end{split}
\end{equation}
and
\begin{equation}\label{eq: wishart_bar_A}
\begin{split}
\bar{A}_{(ij),(kl)}(x) =& x_{ik}(\Lambda\Lambda')_{jl} - q(\sqrt{x}\Theta(x)\sqrt{x})_{ik}(\Lambda\rho\rho'\Lambda')_{jl} + x_{il}(\Lambda\Lambda')_{jk} - q(\sqrt{x}\Theta(x)\sqrt{x})_{il}(\Lambda\rho\rho'\Lambda')_{jk}\\
&+ x_{jk}(\Lambda\Lambda')_{il} - q(\sqrt{x}\Theta(x)\sqrt{x})_{jk}(\Lambda\rho\rho'\Lambda')_{il} + x_{jl}(\Lambda\Lambda')_{ik} - q(\sqrt{x}\Theta(x)\sqrt{x})_{jl}(\Lambda\rho\rho'\Lambda')_{ik}.
\end{split}
\end{equation}
For the given $v$, $D_{(ij)}v = D_{(ji)}v = M_{ij}$ and $D_{(ij),(kl)}v = 0$.  Therefore
\begin{equation}\label{eq: wishart_temp_op_calc}
\begin{split}
&\sum_{i,j,k,l=1}^d A_{(ij),(kl)}D_{(ij),(kl)}v = 0,\\
&\sum_{i,j=1}^d \bar{b}_{ij}D_{(ij)}v = \trace{x\left(K'M + MK - q\zeta'\nu\rho'\Lambda'M - qM\Lambda\rho\nu'\zeta\right)} + \trace{LL'M},
\end{split}
\end{equation}
where we have used repeatedly that $M,X$ are symmetric and that $\trace{ABC} = \trace{BCA} = \trace{CAB}$ for matrices $A,B,C$.  When $d\leq n$, it follows that $\Theta(x) =\idmat{d}$ and  $\bar{A}$ from \eqref{eq: wishart_bar_A} simplifies to
\begin{equation*}
\begin{split}
\bar{A}_{(ij),(kl)}(x) =& x_{ik}\left(\Lambda\Lambda'-q\Lambda\rho\rho'\Lambda'\right)_{jl}+ x_{il}\left(\Lambda\Lambda'-q\Lambda\rho\rho'\Lambda'\right)_{jk}\\
& + x_{jk}\left(\Lambda\Lambda'-q\Lambda\rho\rho'\Lambda'\right)_{il}+ x_{jl}\left(\Lambda\Lambda'-q\Lambda\rho\rho'\Lambda'\right)_{ik},\\
\end{split}
\end{equation*}
and hence using the symmetry for $\Lambda\Lambda'-q\Lambda\rho\rho'\Lambda'$:
\begin{equation}\label{eq: wishart_temp_op_calc_2}
\frac{1}{2}\sum_{i,j,k,l=1}^d \bar{A}_{(ij),(kl)}D_{(ij)}v D_{(kl)}v = 2\trace{x\left(M\Lambda(1-q\rho\rho'\Lambda'M)\right)}.
\end{equation}
Therefore, \eqref{eq: wishart_fv_affine_d_leq_n} follows using \eqref{eq: wishart_b_A_V}, \eqref{eq: wishart_temp_op_calc}, \eqref{eq: wishart_temp_op_calc_2} and the definition of $\mathfrak{F}$ in \eqref{eq: x_ops}. When $d>n$:
\begin{equation*}
\sqrt{x}\Theta(x)\sqrt{x} =\sqrt{x}\left(\sigma'\Sigma^{-1}\sigma\right)(x)\sqrt{x} = x\zeta'\left(\zeta x \zeta'\right)^{-1}\zeta x,
\end{equation*}
thus, using \eqref{eq: wishart_bar_A} it follows that
\begin{equation}\label{eq: wishart_temp_op_calc_3}
\begin{split}
\frac{1}{2}\sum_{i,j,k,l=1}^d \bar{A}_{(ij),(kl)}D_{(ij)}v D_{(kl)}v & = 2\trace{xM\Lambda\Lambda'M} - 2q\trace{x\zeta'\left(\zeta x\zeta'\right)^{-1}\zeta x M \Lambda\rho\rho'\Lambda'M}.
\end{split}
\end{equation}
\eqref{eq: wishart_fv_affine_d_ge_n} now follows from \eqref{eq: wishart_b_A_V}, \eqref{eq: wishart_temp_op_calc} and \eqref{eq: wishart_temp_op_calc_3}.
\end{proof}

\begin{proof}[Proof of Proposition \ref{prop: wishart_good_case}]

Using Lemma \ref{lem: wishart_op_affine_v} it follows for $d\leq n$ that if $M$ solves \eqref{eq: d_leq_n_Ricatti} then $\mathfrak{F}[v] = \lambda$ with $\lambda = \trace{LL'M} + pr_0$. Now, with $D=-M$, \eqref{eq: d_leq_n_Ricatti} takes the form
\begin{equation*}
D\left(2\Lambda(1-q\rho\rho')\Lambda'\right)D  - D(K-q\Lambda\rho\nu'\zeta) - (K-q\Lambda\rho\nu'\zeta)'D - \frac{1}{2}\left(-p(r_1+r_1') + q\zeta'\nu\nu'\zeta\right) = 0.
\end{equation*}
Since the eigenvalues of $\rho\rho'$ are $\rho'\rho$ and $0$, then
\begin{equation*}
2\Lambda(1-q\rho\rho')\Lambda' \geq 2(1-q\rho'\rho)\Lambda\Lambda' > 0.
\end{equation*}
Furthermore, by assumption $-p(r_1+r_1') + q\zeta'\nu\nu'\zeta > 0$.  Thus, the Riccati equation takes the form
\begin{equation}\label{eq: matrix riccati}
D\textbf{B}\textbf{B}'D - D\textbf{A} - \textbf{A}'D - \textbf{C}\textbf{C}' = 0,
\end{equation}
where $\textbf{B} = \sqrt{2\Lambda(1-q\rho\rho')\Lambda'}$, $\bold{A} = K-q\Lambda\rho\nu'\zeta$ and $\textbf{C} = (1/\sqrt{2})\sqrt{-p(r_1+r_1') + q\zeta'\nu\nu'\zeta}$. By \cite[Lemma 2.4.1]{MR1997753}, if there exists matrices $F_1$ and $F_2$ such that $\textbf{A}-\textbf{B}F_1 < 0$ \footnote{Here and in what follows, we write $M<0$ for a given matrix $M\in\md$ with $M+M'<0$.} and $\textbf{A}'-\textbf{C}F_2 < 0$ then there is a unique solution $\hat{M} = -\hat{D}$ to the above such that
\begin{equation}\label{eq: A_B_rel}
\begin{split}
\textbf{A} - \textbf{B}\textbf{B}'\hat{D} &= \textbf{A} + \textbf{B}\textbf{B}'\hat{M}= (K-q\Lambda\rho\nu'\zeta) + 2\Lambda(1-q\rho\rho')\Lambda' \hat{M} < 0.
\end{split}
\end{equation}
Note that $F_1 = \textbf{B}^{-1}\left(\idmat{d} - \textbf{A}\right)$ and $F_2 = \textbf{C}^{-1}\left(\idmat{d} - \textbf{A}'\right)$ are two such matrices. Hence \eqref{eq: matrix riccati} admits a unique solution $\hat{M}$ such that \eqref{eq: A_B_rel} holds.

For $\phi = \hat{v} = \textrm{Tr}(\hat{M}x)$, consider the generator $\cL^{\hat{v}}$ from \eqref{eq: cL_phi_time}, which takes the form
\begin{equation*}
\cL^{\hat{v}} = \frac{1}{2}\sum_{i,j,k,l=1}^d A_{(ij),(kl)}D_{(ij),(kl)} + \sum_{i,j=1}^d\left(\bar{b}_{ij} + \sum_{k,l=1}^d \bar{A}_{(ij),(kl)}\hat{M}_{kl}\right)D_{(ij)}.
\end{equation*}
The drift (i.e. the first order term) above takes the form
\begin{equation*}
\begin{split}
\bar{b}^{ij} &+ \sum_{k,l=1}^d \bar{A}_{(ij),(kl)}\hat{M}_{kl}\\
&= \left(LL' + \left(K-q\Lambda\rho\nu'\zeta + 2\Lambda(1-q\rho\rho')\Lambda'\hat{M}\right)x + x\left(K-q\Lambda\rho\nu'\zeta + 2\Lambda(1-q\rho\rho')\Lambda'\hat{M}\right)'\right)_{ij}\\
&=\left(LL' + (\bold{A}+\bold{B}\bold{B}'\hat{M})x + x(\bold{A}+\bold{B}\bold{B}\hat{M})'\right)_{ij}.
\end{split}
\end{equation*}
Thus, we see that the process $X$ with generator given by $\cL^{\hat{v}}$ is a Wishart process of the form in \eqref{eq: wishart}. Moreover, \eqref{eq: A_B_rel} implies that $K:=\bold{A} + \bold{B}\bold{B}'\hat{M}<0$, hence $X$ is ergodic.  Indeed, $LL' > (d+1)\Lambda\Lambda' > 0$ ensures $X$ does not explode to the boundary of $\sdpos$. Furthermore, consider
\begin{equation*}
u(x) = -\underline{c}\log\left(\det{x}\right) + \overline{c} \norm{x}\eta(\norm{x}),
\end{equation*}
where $\underline{c}, \overline{c}$ are two constants to be determined later, and $\eta(y)$ is a smooth function satisfying $0\leq \eta(y)\leq 1$, $\eta(y) = 1$ for $y>1$ and $0$ for $y<1/2$. Observe that $\lim_{\norm{x}\rightarrow \infty}u(x) =\infty$ and $\lim_{\det(x)\rightarrow 0} u(x) =\infty$, where both limits are uniform as $x$ approaches the boundaries. On the other hand, a calculation similar to that in \cite[Lemmas 5.2 and 5.3]{Robertson-Xing} (with $\bar{\kappa}$ therein equal to $0$) shows the existence of $\underline{c}, \overline{c}, \epsilon>0$ and a sufficiently large sub-domain $E\subset \sdpos$ such that $\cL^{\hat{v}}u(x)\leq -\epsilon$ for all $x\in \sdpos \setminus E$. Therefore \cite[Theorem 6.1.3]{Pinsky} shows that $\prob^{\hat{v}}$ is ergodic. Hence $\hat{v}$ is equal to $\textrm{Tr}(\hat{M}x)$ and $\hat{\lambda} = \textrm{Tr}(LL'\hat{M}) + pr_0$.  This fact follows from \cite[Proposition 2.3]{Robertson-Xing} and \cite[Theorems 2.1,2.2]{Ichihara} which shows the equivalency between $\cL^{\hat{v}}$ being ergodic and $\hat{\lambda}$ being the smallest $\lambda$ with accompanying solution $v$ to $\mathfrak{F}[v]=\lambda$.
\end{proof}

\begin{lem}\label{lem: example_calc}
In the setting of Example \ref{exa: counter-exa}, for $v$ as in \eqref{eq: cand_hatv}, $\mathfrak{F}[v]$ takes the form in \eqref{eq: F_example}.
\end{lem}

\begin{proof}

$\mathfrak{F}[v]$ is given in \eqref{eq: wishart_fv_affine_d_ge_n} of Lemma \ref{lem: wishart_op_affine_v}.  Specifying to the example coefficients and using the representation for $X$,$M$ from \eqref{eq: wishart_x_v_ex}:
\begin{equation*}
\begin{split}
&2M\Lambda\Lambda'M + K'M + MK -q\zeta'\nu\rho'\Lambda'M - qM\Lambda\rho\nu'\zeta + \frac{1}{2}\left(p(r_1+r_1') - q\zeta'\nu\nu'\zeta\right)\\
&\qquad = 2M^2 + 2M - q\rho\nu\left(\begin{array}{c c} 1 & 1\\ 0 & 0\end{array}\right)M - q\rho\nu M\left(\begin{array}{c c} 1 & 0\\ 1& 0\end{array}\right) + pr_1\left(\begin{array}{c c} 1&0\\0&1\end{array}\right) - \frac{1}{2}q\nu^2\left(\begin{array}{c c} 1 & 0 \\ 0 & 0\end{array}\right),\\
&\qquad = 2\left(\begin{array}{c c} M_1^2+M_2^2 & M_2(M_1+M_3)\\ M_2(M_1+M_3) & M_2^2 + M_3^2\end{array}\right) + 2\left(\begin{array}{c c} M_1 & M_2 \\ M_2 &  M_3\end{array}\right) - q\rho\nu\left(\begin{array}{c c} M_1+M_2 & M_2 + M_3\\ 0 & 0 \end{array}\right),\\
&\qquad\qquad - q\rho\nu\left(\begin{array}{c c} M_1+M_2 & 0 \\ M_2 + M_3 & 0\end{array}\right) + pr_1\left(\begin{array}{c c} 1 & 0 \\ 0 & 1\end{array}\right) - \frac{1}{2}q\nu^2\left(\begin{array}{c c} 1 & 0 \\ 0 & 0\end{array}\right),\\
&\qquad = \left(\begin{array}{c c} 2(M_1^2+M_2^2) + 2M_1 - 2q\rho\nu(M_1+M_2) + pr_1 - \tfrac{1}{2}q\nu^2 & 2M_2(M_1+M_3) + 2M_2 - q\rho\nu(M_2+M_3)\\ 2M_2(M_1+M_3) + 2M_2 -q\rho\nu(M_2+M_3) & 2(M_2^2+M_3^2)+2M_3 + pr_1\end{array}\right).
\end{split}
\end{equation*}
Thus,
\begin{equation}\label{eq: wishart_counterex_good}
\begin{split}
&\trace{X\left(2M\Lambda'\Lambda'M + K'M + MK -q\zeta'\nu\rho'\Lambda'M - qM\Lambda\rho\nu'\zeta + \frac{1}{2}\left(p(r_1+r_1') - q\zeta'\nu\nu'\zeta\right)\right)}\\
&\qquad = x\left(2(M_1^2+M_2^2) + 2M_1 - 2q\rho\nu(M_1+M_2) + pr_1 - (1/2)q\nu^2\right)\\
&\qquad\qquad + y\left(4M_2(M_1+M_3) + 4M_2 -2q\rho\nu(M_2+M_3)\right)\\
&\qquad\qquad + z\left(2(M_2^2+M_3^2)+2M_3 + pr_1\right).
\end{split}
\end{equation}
Now, as for the non-constant term on the second line of \eqref{eq: wishart_fv_affine_d_ge_n}, from \eqref{eq: wishart_counter_ex_Theta} we have
\begin{equation}\label{eq: wishart_counterex_bad}
\begin{split}
&-2q\trace{X\zeta'(\zeta X \zeta')^{-1}\zeta X M \Lambda\rho\rho'\Lambda'M}\\
&\qquad = -2q\rho^2\trace{\left(\begin{array}{c c} x &y \\ y &y^2/x\end{array}\right)M\left(\begin{array}{c c} 1&1\\1&1\end{array}\right)M},\\
&\qquad = -2q\rho^2\trace{\left(\begin{array}{c c} x& y\\ y & y^2/x\end{array}\right)\left(\begin{array}{c c} (M_1+M_2)^2 & (M_1+M_2)(M_2+M_3)\\ (M_1+M_2)(M_2+M_3) & (M_2+M_3)^2\end{array}\right)},\\
&\qquad = x\left(-2q\rho^2(M_1+M_2)^2\right)+y\left(-4q\rho^2 (M_1+M_2)(M_2+M_3)\right)+\frac{y^2}{x}\left(-2q\rho^2(M_2+M_3)^2\right).
\end{split}
\end{equation}
Since $\trace{LL'M} + pr_0 = \ell^2(M_1+M_3) + pr_0$, \eqref{eq: F_example} follows from \eqref{eq: wishart_counterex_good} and \eqref{eq: wishart_counterex_bad}.
\end{proof}

\section{Remaining Proofs from Section \ref{sec: converge}}\label{app: B}

\begin{proof}[Proof of Theorem \ref{thm: power}]
 Under Assumptions of Theorem \ref{thm: power}, Statement \ref{stat: long hor} part i) is proved in \cite[Theorems 2.11 and 3.9]{Robertson-Xing}. Note that $\nabla h = \nabla v- \nabla \hat{v}$, part ii) follows from $\nabla h(T, \cdot)\rightarrow 0$ in part i) and the form of $\pi$ in \eqref{eq: pi_v_map}.

 To prove part iii), let us collect two facts from \cite{Robertson-Xing}. First \cite[Proposition 2.3 i)]{Robertson-Xing} implies that $\prob^{\hat{v},x}$, as the solution to the martingale problem for $\cL^{\hat{v}}$, is a well defined probability measure. Therefore discussion after \eqref{eq: Z_phi} proves that $\prob^{\hat{v}, x}$ is equivalent to $\prob^x$ on $\F_t$ for any $t\geq 0$. Second,
\begin{equation}\label{eq: hatv_quad_var_lim}
 \lim_{T\rightarrow \infty} \expec^{\prob^{\hat{v},x}}\bra{\int_0^t
    \sum_{i,j,k,l=1}^d D_{(ij)}h\bar{A}_{(ij),(kl)}D_{(kl)} h (T-u,X_u)\,du} =  0.
\end{equation}
Indeed,
since the integrand in \eqref{eq: hatv_quad_var_lim} is independent of the Brownian motion $W$, \eqref{eq: hatv_quad_var_lim} is proved in \cite[Theorems 2.9 and 3.9]{Robertson-Xing}.

Let us use the previous two facts to prove \eqref{eq: dist strategy} first. To this end, using \eqref{eq: pi_v_map}, we obtain in either cases $m\geq n$ or $m<n$,
 \begin{equation*}
\begin{split}
&\left(\pi(T-t,x;v) - \pi(x;\hat{v})\right)'\Sigma(x)\left(\pi(T-t,x;v) -
  \pi(x;\hat{v})\right),\\
&\qquad  = \frac{1}{(1-p)^2}\left(\sum_{i,j,k,l=1}^d D_{(ij)}h
\rho'(a^{ij})'C'\Theta C a^{kl}\rho D_{(kl)}h\right)(T-t,x),\\
&\qquad \leq \frac{1}{(1-p)^2}\left(\sum_{i,j,k,l=1}^d D_{(ij)}h
\trace{a^{ij}(a^{kl})'} D_{(kl)}h\right)(T-t,x),\\
&\qquad \leq \frac{1}{\uk(1-p)^2}\left(\sum_{i,j,k,l=1}^d
  D_{(ij)}h\bar{A}_{(ij),(kl)} D_{(kl)}h\right)(T-t,x),
\end{split}
\end{equation*}
where the first inequality follows from \eqref{eq: bar_A_to_A_compare_5} and the second inequality follows from the first inequality in \eqref{eq: A_barA_compare}. Then \eqref{eq: hatv_quad_var_lim} yields
\[
 \lim_{T\rightarrow \infty} \expec^{\prob^{\hat{v},x}}\bra{\int_0^t\left(\pi^T_u - \hat{\pi}_u\right)'\Sigma(X_u)\left(\pi^T_u -
  \hat{\pi}_u\right)\, du}=0.
\]
This implies the convergence in probability $\prob^{\hat{v},x}$, hence in $\prob^x$, since $\prob^{\hat{v},x}$  is equivalent to $\prob^x$ on $\F_t$.

To prove \eqref{eq: ratio wealth power},
 apply the first identity of \eqref{eq: prop_1_1}, where we choose $\phi=v$ from Proposition \ref{prop: v wellposed} and $\pi= \pi^T$ from \eqref{eq: opt_strat_T}. Taking difference of this identity when $t=t$ and $t=0$ respectively yields
 \[
  \pare{\frac{\W^T_t}{w}}^p = Z^{v, T}_t e^{v(T, x) - v(T-t, X_t)}.
 \]
 On the other hand, apply the first identity of \eqref{eq: prop_1_1} again, but choose $\pi= \hat{\pi}$ from \eqref{eq: opt_strat} and $\phi(t,x)= \hat{\lambda} t + \hat{v}(x)$, where $(\hat{\lambda}, \hat{v})$ comes from Proposition \ref{prop: ergodic wellposed} and the current choice of $\phi$ satisfies $\phi_t = \fF[\phi]$ due to \eqref{eq: v_ergodic}. Taking difference of this identity when $t=t$ and $t=0$ respectively, we obtain
 \[
  \pare{\frac{\hat{\W}_t}{w}}^p = Z^{\hat{v}}_t e^{\hat{\lambda} T+ \hat{v}(x) - \hat{\lambda}(T-t) - \hat{v}(X_t)}.
 \]
 Therefore, the ratio between the previous two identities reads
 \begin{equation}\label{eq: ratio wealth prowr resp}
  \frac{\W^T_t}{\hat{\W}_t} = \pare{\frac{Z^{v,T}_t}{Z^{\hat{v}}_t} e^{h(T, x) - h(T-t, X_t)}}^{\frac1p},
 \end{equation}
 where $h$ is defined in Statement \ref{stat: long hor} part i).
 It has been proved in part i) that $h(T, \cdot)\rightarrow C$ for some constant $C$. Therefore $e^{h(T,x) - h(T-t,X_t)} \rightarrow 1$ a.s. as $T\rightarrow \infty$.
 In the next paragraph, we will show
 \begin{equation}\label{eq: ratio Z}
  \prob^{\hat{v}, x}-\lim_{T\rightarrow \infty} \frac{Z^{v,T}_t}{Z^{\hat{v}}_t} =1.
 \end{equation}
 Plugging the previous two convergence back into \eqref{eq: ratio wealth prowr resp}, it follows \[\prob^{\hat{v},x}-\lim_{T\rightarrow \infty} \frac{\W^T_t}{\hat{\W}_t} =1.\]
 Recall from Remark \ref{rem: numeraire} that $\W^T/\hat{\W}$ is a $\prob^{\hat{v},x}$-supermartingale. Combining the previous convergence with Scheff\'{e}'s  lemma, we obtain
 \[
  \lim_{T\rightarrow \infty} \expec^{\prob^{\hat{v},x}}\bra{\left|\frac{\W^T_t}{\hat{\W}_t}-1\right|}=0,
 \]
 Applying \cite[Lemma 3.9]{guasoni.al.11} under $\prob^{\hat{v},x}$, the previous convergence then yields
 \[
  \prob^{\hat{v},x}-\lim_{T\rightarrow \infty} \sup_{0\leq u\leq t}\left|\frac{\W^T_u}{\hat{\W}_u}-1\right|=0.
 \]
 Hence \eqref{eq: ratio wealth power} is confirmed after utilizing the equivalence between $\prob^{\hat{v},x}$ and $\prob^x$.

 It remains to prove \eqref{eq: ratio Z}. To this end, using \eqref{eq: Z_phi} for $v$ and $\hat{v}$, and the definition of $h$, it follows that $Z^{v,T}_t/Z^{\hat{v}}_t = \mathcal{E}(L^T)_t$, where the $\prob^{\hat{v}, x}$-local martingale $L^T$ takes the form
 \begin{equation*}
\begin{split}
L^T_t = &\int_0^t \sum_{k,l=1}^d d\hat{B}^{kl}_u\left(\sum_{i,j=1}^d \left(a^{ij}_{kl} - q(C'\Theta C a^{ij}\rho)_k\rho_l\right)D_{(ij)}h \right)(T-u,X_u)\\
& + \int_0^t \sum_{k=1}^m d\hat{W}^k_u\left(-q\sum_{i,j=1}^d (D'\Theta C a^{ij}\rho)_k D_{(ij)}h\right) (T-u,X_u),\qquad t\leq T,
\end{split}
\end{equation*}
where $\hat{B}$ and $\hat{W}$ are $\prob^{\hat{v},x}$ independent $\md$ and $\reals^m$ dimensional Brownian motions. Calculation using $\rho'\rho CC' + DD' = 1_m$ and $\Theta\Theta = \Theta$ shows that
\begin{equation*}
[L^T,L^T]_t = \int_0^t \left(\sum_{i,j,k,l=1}^d D_{(ij)}h\left(\bar{A}_{(ij),(kl)} - q(1-q)\rho'(a^{ij})'C'\Theta C a^{kl}\rho\right)D_{(kl)}h\right)(T-u,X_u)du.
\end{equation*}
Using \eqref{eq: bar_A_to_A_compare_5} at $\theta = D h\in\sd$ it follows for $p < 0$ ($0<q<1$) that
\begin{equation*}
[L^T,L^T]_t \leq  \int_0^t \left(\sum_{i,j,k,l=1}^d D_{(ij)}h \bar{A}_{(ij),(kl)}D_{(kl)}h\right)(T-u,X_u)du,
\end{equation*}
and for $0 < p < 1$ ($q<0$) that
\begin{equation*}
\begin{split}
[L^T,L^T]_t &\leq  \int_0^t \left(\sum_{i,j,k,l=1}^d D_{(ij)}h\left(\bar{A}_{(ij),(kl)} - q(1-q)\trace{a^{ij}(a^{kl})'}\right)D_{(kl)}h\right)(T-u,X_u)du\\
&\leq  \pare{1-\frac{q(1-q)}{\uk}}\int_0^t \left(\sum_{i,j,k,l=1}^d D_{(ij)}h\bar{A}_{(ij),(kl)}D_{(kl)}h\right)(T-u,X_u)du
\end{split}
\end{equation*}
where the last inequality uses Lemma \ref{lem: barA}. From \eqref{eq: hatv_quad_var_lim} it thus follows that
\begin{equation*}
\lim_{T\uparrow\infty} \expec^{\prob^{\hat{v},x}}\bra{[L^T,L^T]_t} = 0,
\end{equation*}
which implies $\prob^{\hat{v},x}-\lim_{T\rightarrow \infty}[L^T,L^T]_t=0$. Combining the previous convergence and the fact that $L^T$ is continuous local martingales, it follows $\prob^{\hat{v},x}-\lim_{T\rightarrow \infty} \mathcal{E}(L^T)_t=1$, hence \eqref{eq: ratio Z} holds.
\end{proof}

\begin{proof}[Proof of Theorem \ref{thm: turnpike}]
Given results in \cite[Theorems 2.9 and 3.9]{Robertson-Xing}, the statement follows from the same argument in \cite[Theorem 2.9]{guasoni.al.11}. We now check that the assumptions in \cite{guasoni.al.11} are satisfied in the current setting. First, for each $T>0$, there exists a probability measure $\qprob^{T,x}$ such that $\qprob^{T,x}$ is equivalent to $\prob^x$ on $\F_T$ and such that $e^{-\int_0^\cdot r(X_u)du}S$ is a $\qprob^{T,x}$-local martingale on $[0,T]$. Indeed, let $\theta : \sdpos\mapsto \reals^k$ be a continuous function and set
\begin{equation*}
Z_t = \mathcal{E}\left(-\int_0^\cdot \sum_{k=1}^d \theta_k (X_u)dW^k_u\right)_t,
\end{equation*}
The continuity of $\theta$ and the $\prob$ independence of $X$ and $W$ ensure that $Z$ is also a $\prob^x$-martingale, cf. \cite[Lemma 4.8]{Karatzas-Kardaras}. Under Assumption \ref{ass: rho_strong} we may choose
$\theta = D'(DD')^{-1}\sigma'\nu$,
and it follows that $\theta$ is continuous.  Since $Z$ is a $\prob^x$-martingale, for each $T$ we may define a probability $\qprob^{T,x}$, which is equivalent to $\prob^x$ on $\F_T$, via $d\qprob^{T,x}/d\prob^x |_{\F_T} = Z_T$. Using Girsanov's theorem, a direct calculation shows that $e^{-\int_0^\cdot r(X_u) du}S$ is $\qprob^{T,x}$-local martingale. Therefore \cite[Assumption 2.3]{guasoni.al.11} is satisfied. On the other hand, Propositions \ref{prop: v wellposed} and \ref{prop: verification} combined implies that the value of the optimization problem in \eqref{eq: power op} is finite for all $T\geq 0$. Therefore \cite[Assumption 2.4]{guasoni.al.11} is satisfied as well. On the other hand, Assumptions \ref{ass: ratio} and \ref{ass: grow} are exactly \cite[Assumptions 2.1 and 2.2]{guasoni.al.11} respectively.

Therefore \cite[Proposition 2.5]{guasoni.al.11} proves that, for  all $\eps > 0$,
\begin{equation}\label{eq: PT_turnpikes}
\begin{split}
\lim_{T\uparrow\infty}& \prob^{v,T,x}\bra{ \sup_{u\leq t}\left|\frac{\W^{1,T}_u}{\W^T_u} - 1\right| \geq \eps} = 0,\\
\lim_{T\uparrow\infty}& \prob^{v,T,x}\bra{ \int_0^t \left(\pi^{1,T}_u - \pi^{T}_u\right)'\Sigma(X_u)\left(\pi^{1,T}_u - \pi^{T}_u\right) du \geq \eps} = 0.
\end{split}
\end{equation}
Here since the martingale problem for $\cL^{v,T-\cdot}$ is well-posed, cf. \cite[Lemma 4.1]{Robertson-Xing}, $\prob^{T,v,x}$ is defined via \eqref{eq: Z_phi} with $\phi=v$. From the definitions of $\prob^{v,T,x}$ and $\prob^{\hat{v},x}$, it follows
\[
 \left.\frac{d\prob^{v,T,x}}{d\prob^{\hat{v},x}}\right|_{\F_t} = \frac{Z^{v,T}_t}{Z^{\hat{v}}_t}.
\]
Note that both events on the left-hand-side of \eqref{eq: PT_turnpikes} are $\F_t$-measurable. Therefore, \eqref{eq: ratio Z} implies \eqref{eq: PT_turnpikes} holds when $\prob^{v,T,x}$ is replaced by $\prob^{\hat{v},x}$, hence also by $\prob^x$, since $\prob^{\hat{v},x}$ and $\prob^x$ are equivalent on $\F_t$. Lastly, the extension to Statement \ref{stat: turnpike} is immediate after utilizing Statement \ref{stat: long hor} part iii).
\end{proof}

\begin{proof}[Proof of Proposition \ref{prop: wishart}]
 Let us verify Assumption \ref{ass: coeff_master_list} is satisfied under the parameter restrictions of this proposition. Then the statements readily follow from Theorems \ref{thm: power} and \ref{thm: turnpike}. First, for the Wishart factor model described in Section \ref{subsec: wishart}:
 \begin{align*}
  & V(x) = pr_0 + \frac12 \trace{\pare{x(p(r_1+r_1') -q \zeta' \nu \nu' \zeta(x))}},\\
  & \overline{b}(x) = LL' + \overline{K}(x) x+ x \overline{K}(x)',
 \end{align*}
 where $\overline{K} = K-q \Lambda \rho \nu' \zeta(x)$. Since $\rho, \nu, \zeta$ are bounded, it is clear that $\overline{b}$ has at most linear growth. We have seen from Example \ref{ex: wishart} that $f(x)=x$ and $g(x)= \Lambda \Lambda'$. Then $\trace{f(x)}\trace{g(x)} =\trace{x} \trace{\Lambda\Lambda'}\leq \sqrt{d} \trace{\Lambda \Lambda'} \norm{x}$. In particular, $\alpha_1$ in Assumption \ref{ass: coeff_master_list} part 2) can be chosen as $\sqrt{d} \trace{\Lambda \Lambda'}$. To see the previous inequality, let $(\lambda_i)_{i=1, \dots, d}$ be eigenvalues of $x$, then Cauchy-Schwarz inequality yields $\trace{x} = \sum_{i=1}^d \lambda_i \leq \sqrt{d} (\sum_{i=1}^d \lambda_i^2)^\frac12 = \sqrt{d} \norm{x}$. To verify Assumption \ref{ass: coeff_master_list} part 3), we choose $-\beta_1$ to be larger than any largest eigenvalue of $(\overline{K}+\overline{K}')(x)$ for $x\in \sdpos$. Since $\overline{K}(x)$ is bounded on $\sdpos$, its largest eigenvalue is uniformly bounded on $\sdpos$. Move on to Assumption \ref{ass: coeff_master_list} part 4). When $0<p<1$, $q<0$, then
 \[
  -|p r_0| - \frac12 \norm{p(r_1+ r_1') - q \zeta' \nu \nu' \zeta(x)}\norm{x} \leq V(x) \leq |p r_0| + \frac12 \norm{p(r_1+ r_1') - q \zeta' \nu \nu' \zeta(x)}\norm{x}, \quad x\in \sdpos.
 \]
 Hence we can choose $-\gamma_1 = \gamma_2 =  (1/2)\sup_{x\in \sdpos} \norm{p(r_1+ r_1') - q \zeta' \nu \nu' \zeta(x)}$. When $p<0$, $q>0$, then
 \[
  -|p r_0| - \frac12 \norm{p(r_1+ r_1') - q \zeta' \nu \nu' \zeta(x)}\norm{x} \leq V(x) \leq |p r_0| - \lambda_{min}(x) \norm{x},
 \]
 where $\lambda_{min}(x)$ is the smallest eigenvalue of $(1/2)(-p(r_1+ r_1') + q \zeta' \nu \nu' \zeta(x))$. Hence we can choose the same $\gamma_2$ as above, but $\inf_{x\in \sdpos} \lambda_{min}(x)$ as $\gamma_1$. Therefore Assumption \ref{ass: coeff_master_list} part 4) is verified.

 Let us now check part 5). When $p<0$, because $r_1+r_1'\geq 0$ and $\zeta' \nu \nu' \zeta\geq 0$, $\lambda_{min}(x)\geq 0$ for any $x\in \sdpos$, then $\gamma_1 \geq 0$. When $(\overline{K} + \overline{K}')(x)\leq -\epsilon \idmat{d}$ for any $x\in \sdpos$, $\beta_1 >0$, hence part 5)-iii) is satisfied. When $-p(r_1 + r_1') + q(\zeta' \nu \nu' \zeta(x)) \geq \epsilon \idmat{d}$ for any $x\in \sdpos$, $\gamma_1>0$, hence we are in part 5)-i). In such a case, $\trace{f(x) x g(x) x} = \trace{x^3 \Lambda \Lambda'} \geq \alpha_2 \norm{x}^3$ for some $\alpha_2>0$, where the inequality holds due to $\Lambda \Lambda' >0$. When $0<p<1$, then $\gamma_1 = -(1/2)\sup_{x\in \sdpos} \norm{p(r_1+ r_1') - q \zeta' \nu \nu' \zeta(x)}<0$. Recall $\ok=1-q$ from Lemma \ref{lem: barA} and $\alpha_1= \sqrt{d} \trace{\Lambda \Lambda'}$ from part 2), then \eqref{eq: p>0 cond} is equivalent to $\beta_1^2 + 16 \ok \alpha_1 \gamma_1>0$ from part 5)-ii). Therefore, Assumption \ref{ass: coeff_master_list} part 5) is satisfied as well.

Finally, let us verify part A)-C). For A), calculation shows that
\begin{equation*}
H_\epsilon(x; \overline{b}) = \trace{(LL' - (1+d + \epsilon) \Lambda \Lambda')x^{-1}} + 2 \trace{\overline{K}(x)}.
\end{equation*}
Then $LL' > (d+1)\Lambda \Lambda'$ ensures the existence of $\epsilon>0$ such that $LL'-(1+d+\epsilon)\Lambda \Lambda' > 0$. Hence the previous inequality and the assumption that $\overline{K}$ is bounded on $\sdpos$ implies that $\inf_{x\in \sdpos} H_\epsilon(x; \overline{b})>-\infty$. As for B), part A) implies the existence of $\delta>0$ such that $H_\epsilon(x; \overline{b}) \geq \delta \trace{x^{-1}} + 2\trace{\overline{K}(x)}$. Observe that, for any $c_0>0$, $\delta \trace{x^{-1}} + c_0 \log(\det{x})\rightarrow \infty$ as $\det{x}\downarrow 0$. Then part B) is confirmed. Lastly, for part C), there exist $\delta, C>0$ such that $H_0(x, \overline{b}) + c_1 V(x) \geq \delta \trace{x^{-1}} - \gamma_2 \norm{x} + C$, which goes to $\infty$ as $\det{x} \downarrow 0$. This concludes verification of all parameter restrictions in Assumption \ref{ass: coeff_master_list}.
\end{proof}

\bibliographystyle{siam}
\bibliography{biblio}

\end{document}